\documentclass{article}
\usepackage[sectionbib]{natbib}
\usepackage{float}

\usepackage[english]{babel}

\usepackage[letterpaper,top=2cm,bottom=2cm,left=3cm,right=3cm,marginparwidth=1.75cm]{geometry}

\usepackage{amsmath, amsthm, amssymb, amsfonts}
\usepackage{graphicx}
\usepackage{mathtools}
\usepackage[colorlinks=true, allcolors=blue]{hyperref}
\usepackage{bm}
\usepackage{bbm}
\usepackage{soul}
\usepackage{verbatim}

\newtheorem{theorem}{Theorem}
\newtheorem{definition}{Definition}
\newtheorem{assumption}{Assumption}

\DeclarePairedDelimiter\abs{\lvert}{\rvert}%
\makeatletter
\let\oldabs\abs
\def\abs{\@ifstar{\oldabs}{\oldabs*}}

\title{Privacy Amplification for Synthetic data using Range Restriction}
\author{Jingchen Hu\footnote{mhu7@binghamton.edu} \\ Binghamton University
\and Matthew R. Williams \\ RTI International
\and Terrance D. Savitsky \\U.S. Bureau of Labor Statistics}

\begin{document}
\maketitle
\begin{abstract}
 We introduce a new class of range restricted formal data privacy standards that condition on owner's beliefs about sensitive data ranges. By incorporating this additional information, we can provide a stronger privacy guarantee (e.g., an amplification).  The range-restricted formal privacy standards protect only a subset (or ball) of data values and exclude ranges (or balls) believed to be already publicly known. The privacy standards are designed for the risk-weighted pseudo posterior mechanism (PPM) used to generate synthetic data under an asymptotic differential (aDP) privacy guarantee. The PPM downweights the likelihood contribution for each record proportionally to its disclosure risk.  The PPM is adapted under inclusion of beliefs by adjusting the risk-weighted pseudo likelihood.  We introduce two alternative adjustments. The first expresses data owner's knowledge of the sensitive range as a probability, $\lambda$, that a datum value drawn from the underlying generating distribution lies \emph{outside} the ball or subspace of values that are sensitive. The portion of each datum likelihood contribution deemed sensitive is then $(1-\lambda) \leq 1$ and is the only portion of the likelihood subject to risk down-weighting.  The second adjustment encodes knowledge as the difference in probability masses $P(R) \leq 1$ between the edges of the sensitive range, $R$. We use the resulting \emph{conditional} pseudo likelihood for a sensitive record, which boosts its worst case tail values away from 0.  We compare privacy and utility properties for the PPM under the aDP and range restricted privacy standards through a series of simulation studies and a real data application to an accelerated life testing dataset.
\end{abstract}

\section{Introduction}\label{sec:introduction}

Differential privacy (DP) \citep{Dwork:2006:CNS:2180286.2180305} offers a mathematically verifiable privacy standard and associated privacy guarantee that are used to regulate the influence of any datum in a class of databases on an estimator of interest (under a randomized mechanism) as a means to regulate its privacy exposure. 
Statistical models, called synthesizers \citep{Rubin1993synthetic, Little1993synthetic, Drechsler2011book}, can be used as randomized mechanisms for generation of synthetic under a DP guarantee. The DP guarantee is based on a worst case sensitivity of the synthesizer to inclusion of a datum over the space of all datasets of a class and the model parameter space. Synthesizers equipped with a DP guarantee often induce a high degree of distortion into data that can render the utility of the synthetic data to be of low quality \citep{Montoya_Perez_2024, Bowen_Snoke_2021}.  

Novel privacy standards have been developed, which are often extensions of DP, that maintain the mathematical verification property of DP, but ``relax" its worst case formulation of the guarantee.
\citet{Dwork:2006:CNS:2180286.2180305} propose a $2-$parameter, $(\epsilon,\delta)$ relaxation of the DP standard where $\delta$ denotes the probability that there exists a dataset where the $\epsilon$ guarantee will be exceeded.  The $1-$parameter DP standard is recovered as $\delta\downarrow 0$. A more recent $2-$parameter privacy standard that relaxes DP is the $(\alpha,\epsilon)-$R\'{e}nyi DP standard of \citet{mironov2019renyi} and \citet{Mironov2017RnyiDP} based on the $\alpha > 1$ R\'{e}nyi divergence which is monotonic in $\alpha$ \citep{6832827}.  DP is recovered as $\alpha\uparrow\infty$.  \citet{SavitskyWilliamsHu2020ppm} develop a mathematically provable standard that is guaranteed for the observed dataset and extends to the space of all datasets asymptotically.  They label their guarantee as ``asymptotic" differential privacy (aDP).  The aDP standard is based on the Bayesian posterior distribution used to produce synthetic data equipped with an aDP guarantee.  

\citet{SavitskyWilliamsHu2020ppm} develope a disclosure risk-weighted pseudo posterior mechanism that surgically distorts only portions of the closely-held data distribution that express disclosure risks to efficiently preserve utility of the synthetic data.  The pseudo posterior mechanism under aDP is reviewed in Section \ref{sec:intro:PPMreview}.

What these relaxed privacy standards all have in common is that they trade-off the strength of the privacy guarantee for improved utility.

A relatively recent approach to tune the trade-off between the strength of the privacy guarantee and the resulting utility (of the synthetic data) incorporates information known about the data by a putative intruder or the interested public to reduce the sensitive range of the data requiring privacy protection. These approaches may be viewed as ``amplifications" of the DP guarantees on which they are based because they use known information to sharpen the guarantee to improve the trade-off with utility.  The Pufferfish privacy (PP) \citep{e2c5145465454699a3b17fecbe979a05, song2017pufferfishprivacymechanismscorrelated} admits a collection of privacy standards based on a definition of secrets and the nature of information assumed in possession of the intruder.  It bears mention that PP is generally used under additive noise mechanisms such as those utilized to protect summary and tabular data statistics with a DP guarantee, rather than to produce synthetic data equipped with a DP guarantee, which is the focus of this paper.

This paper contributes a novel extension of the aDP guarantee of \citet{SavitskyWilliamsHu2020ppm} that incorporates knowledge of the interested public possessed by the owner of the closely-held data to amplify and sharpen the privacy guarantee.  The knowledge is incorporated by defining sensitive ranges around the datum value for each record that need to be protected with a privacy guarantee. If the datum value lies outside of the sensitive range, then it does not need to be protected.  Amplification occurs by limiting the worst case bound to a subset of the space of data defined by the ranges, rather than using the entire space of the data.  The range-restricted privacy extensions developed in this paper effectively combine the relaxed aDP guarantee with our new range-restricted amplification.  A similar extension of PP is performed under additive noise mechanisms by incorporating R\'{e}nyi DP with PP \citep{pierquin2024rnyi}.    

\subsection{Review of the pseudo posterior mechanism from \citet{SavitskyWilliamsHu2020ppm}}
\label{sec:intro:PPMreview}
The new method introduced in this paper incorporates knowledge of non-sensitive data ranges to sharpen or amplify the aDP guarantee under a pseudo posterior mechanism that produces synthetic data.  We review the pseudo posterior mechanism under aDP to set context for our privacy amplifications introduced in Sections~\ref{sec:averaged} and \ref{sec:truncated}.  We begin by specifying the $\bm{\alpha}-$weighted pseudo posterior mechanism $\mathcal{M}$ as the pseudo posterior distribution.  The pseudo posterior distribution, in turn, is achieved by multiplying a pseudo likelihood with a prior distribution specified for model parameters,
\begin{equation}
\xi^{\bm{\alpha}(\bm x)}(\theta \mid \bm x) \propto \prod_{i=1}^{n} p( x_i \mid \theta)^{\alpha_i} \times \xi(\theta).
\end{equation}
The pseudo posterior distribution first begins by specifying any Bayesian synthesis model for confidential data $\bm x$. A formal privacy guarantee is encoded into the specified model posterior distribution by formulating risk-based weight, $\alpha_i \propto 1/\max_{\theta \in \{\xi^{\bm{\alpha}(\bm x)}(\theta \mid \bm x)\}_{m}}\abs{f_{\theta}(x_{i})}$, where log-likelihood $f_{\theta}(x_{i}) = \log p_{\theta}(x_{i})$.  A data record with a large absolute log-likelihood value is typically located in the distribution tail and is relatively risky because it is surrounded by few units that express similar data values.  

An unweighted posterior distribution is first estimated and the risk-based weight, $\alpha_{i}$, is computed as the maximum over posterior draws of parameters, $\theta$, for datum, $x_{i}$.  The vector of risk-based weights, $\bm{\alpha} = (\alpha_{1},\ldots,\alpha_{n})$, are normalized such that each $\alpha_{i}\in[0,1]$.

The pseudo posterior mechanism is then constructed by down-weighting each likelihood component by $\alpha_i \in [0, 1]$ where the higher the disclosure risk, the lower is $\alpha_i$.  The down-weighting of each likelihood component by its relative disclosure risk is a surgical treatment that concentrates down-weighting to the risky portions of the closely-held data distribution.

The connection between the pseudo posterior mechanism and DP is through the sensitivity or Lipschitz upper bound on the logarithm of the weighted log-likelihood,
\begin{align}
\Delta_{\bm{\alpha},\mathbf{x}} &= \max_{\theta \in \{\xi^{\bm{\alpha}(\bm x)}(\theta \mid \bm x)\}_{m}} \max_{i\in 1,\ldots,n} | \alpha_{i} \times f_{\theta_m}(x_{i}) |,
\end{align}
where the maximum over $\theta$ is taken with respect to (draws indexed by $m$ from) the posterior distribution for each record $i$. Each posterior predictive draw of synthetic data, $\bm{x}^{\ast} \sim p(x \mid \theta^{\ast})$ under the pseudo posterior mechanism $\theta^{\ast} \sim \xi^{\bm{\alpha}(\bm x)}(\theta \mid \bm x)$ is equipped with a $\epsilon_{\bm x} = 2 \Delta_{\bm{\alpha},\mathbf{x}}$ that is \emph{local} to database $\bm{x}$, rather than global over the space of databases, $\bm{x}^{n}\in\mathcal{X}^{n}$.

\citet{SavitskyWilliamsHu2020ppm} show that the local $\Delta_{\bm{\alpha},\mathbf{x}}$ contracts onto the global $\Delta_{\bm{\alpha}}$, asymptotically in sample size, which in turn drives the contraction of $\epsilon_{\mathbf{x}}$ onto $\epsilon$.  For a sample size $n$ sufficiently large, $\epsilon_{\mathbf{x}} = \epsilon$.  More formally, the authors demonstrate that the local Lipschitz satisfies a relaxed form of DP that they label aDP where a stands for asymptotic.  

The rest of the paper is organized as follows: 
We start with describing the range-averaged standard in Section \ref{sec:averaged}, followed by the range-truncated standard in Section \ref{sec:truncated}. Section \ref{sec:simulation} contains a series of extensive simulation studies to illustrate the privacy amplification effects of the two novel range-restricted privacy standards as well as their differences. We apply the new standards to an accelerated life testing dataset in Section \ref{sec:application} and end with a few concluding remarks in Section \ref{sec:conclusion}.

\section{Probability of record inclusion in sensitive range}\label{sec:averaged}

The owner of the closely-held data may offer judgment about a subset range or region around each datum that they deem to be sensitive for re-identification that requires privacy protection.  Their judgment may derive from knowledge of the data user communities, such as policy-makers and researchers.  Suppose the owner of the closely-held data defines a width with respect to sensitive variable, $\mathbf{x}$, as a percentile of the distribution for each $x_{i},~i \in (1,\ldots,n)$.  We implement this idea by constructing percentile radius, $r = [a,b]$; for example, suppose $r = [0.3,1.7]$ is selected by the owner of the closely-held data and and further suppose that $x_{i} = \$100,000$ represents the household income for household $i$.  Then, we say that the interested public knows that the income for household $i$ lies inside  of the interval $[a \times x_{i},b \times x_{i}] = [\$30,000, \$170,000]$.  In other words, the interested public knows that the income for household $i$ is greater than $\$30,000$ and less than $\$170,000$. For simplicity of exposition, we will use this ``radius" idea; however other methods for forming intervals (e.g. pre-defined income brackets) could also be used. We propose to restrict privacy protection to the portion of the data support \emph{inside} the sensitive interval where the intruder does not know where lies the confidential income value, $x_{i}$, for household $i$.  By the complement of the interval on the support of $x_i$ we choose to protect against privacy disclosure, we do \emph{not} consider values outside of this interval to be sensitive because the intruder already knows that the income value must lie within this interval.   Our approach proposes an \emph{amplified}  privacy standard through leveraging or conditioning on known information to target privacy protection to a subspace of the data, rather than to the entire space.

More formally, we propose to measure the known information on sensitive ranges by constructing the following event probability, 
\begin{equation}\label{nonprivaterange}
    \lambda_{i} = \Pr_{\mathcal{M},a,b}\left(x^{\ast}_{i} \notin [a \times x_{i},b \times x_{i}]\right)
\end{equation}
from the posterior predictive distribution, $p_{\mathcal{M}}(\mathbf{x}^{\ast}\mid\mathbf{x})$ where we utilize the same model, $\mathcal{M}$, estimated in the first step to construct the risk-based weights, $\alpha_{i}$, for formulating our pseudo posterior privacy mechanism.  The posterior predictive distribution under $\mathcal{M}$ is used to estimate the probability that under repeated data generation that each datum will lie \emph{outside} the radius such that it does \emph{not} incur a privacy loss (because it lies outside the selected sensitive range).   Each $\lambda_{i}$, then, represents the probability that a realized  datum value for unit $i$ does \emph{not} need to be protected because the interested public already knows that the confidential datum value does not lie in the space outside of the ball/interval.  By using the distribution of realized values for a given record under repeated sampling, we vary the degree of protection required for that record depending on how much of its generating distribution will fall within the sensitive range.

One may express the amount of information under model $\mathcal{M}$ about the data
that is deemed known or non-sensitive and, therefore, not requiring protection as, 
\begin{equation}
p_{\bm{\theta}}^{\bm{\lambda}}(\mathbf{x}) = \mathop{\prod}_{i=1}^{n}p(x_{i}\mid\bm{\theta})^{\lambda_{i}}.
\end{equation}
We next turn to adjusting our pseudo posterior mechanism under the set-up where only $(1-\lambda_{i})$ fraction of datum $i$ requires protection. One may derive a revised likelihood weight $\alpha^{\ast}_{i}$ by applying risk-based weight, $\alpha_{i}$ to the $(1-\lambda_{i})$ portion of the the likelihood contribution for datum $i$ that is sensitive. This result is achieved by decomposing the pseudo likelihood of the final pseudo posterior mechanism into the product of two components 
\begin{equation}\label{eq:decomp}
p_{\bm{\theta}}^{\bm{\alpha}{\ast}}(\mathbf{x}) = p_{\bm{\theta}}^{\bm{\lambda}^c \bm{\alpha}}(\mathbf{x}) p_{\bm{\theta}}^{\bm{\lambda}}(\mathbf{x})
= \mathop{\prod}_{i=1}^{n}p(x_{i}\mid\bm{\theta})^{(1-\lambda_i) \alpha_i}
\mathop{\prod}_{i=1}^{n}p(x_{i}\mid\bm{\theta})^{\lambda_{i}},
\end{equation}
where $\bm{\lambda}^{c}$ denotes the ``complementary" subspace of the data that is deemed sensitive.

Our likelihood decomposition into sensitive and non-sensitive components implies the following revised weight, $\alpha^{\ast}_{i}$, 

\begin{equation}
 \alpha^{\ast}_{i} = \lambda_i +  (1-\lambda_i) \times \alpha_i,    
\end{equation}
that replaces weight $\alpha_{i}$ in the second modeling step for estimation of the pseudo posterior distribution.  

We have designed $\alpha^{\ast}_{i}$ to maintain a pseudo posterior mechanism that downweights each unit $i$ based on its privacy risk in a manner that accounts for the information conveyed in $\lambda_{i}$.  When a unit is highly risky as measured by the absolute value of its log-likelihood under the model $\mathcal{M}$ estimated on the full data such that $\alpha_{i} = 0$ (is set to $0$), then $\alpha^{\ast}_{i} = \lambda_{i}$ because $\lambda_{i}$ represents the portion of the data support that does \emph{not} need to be privacy protected such that it is the minimum (exponentiated) fraction of information presented the model.  By contrast, when datum $i$ presents no disclosure risk such that $\alpha_{i} = 1$, then $\alpha^{\ast}_{i} = 1$, indicating that we present the full datum (rather than a fraction) to the estimation model because there is no privacy protection required under this scenario. Of course, when $\lambda_{i} = 0$ it is supposed that the interested public has \emph{no} knowledge about a subspace of the private variable and we return to the usual case where $\alpha^{\ast}_{i} = \alpha_{i}$. 

We may view $\alpha^{\ast}_{i}$ as an adjusted version of $\alpha_{i}$ that ``adds in" known information conveyed by the restricted sensitive data range; for example, the minimum of $\alpha^{\ast}_{i}$ is always $\lambda_{i}$.

Since the second component of Equation~(\ref{eq:decomp}) is viewed as acceptable to expose to the public, we only use the first component to assess privacy risk. This leads us to a revised sensitivity of our mechanism for dataset $\mathbf{x}$ (see \citet{SavitskyWilliamsHu2020ppm}),

\begin{align}
\Delta_{\bm{\alpha},\bm{\lambda},\mathbf{x}} &= \max_{\theta \in \{\xi^{\bm{\alpha}^{\ast}(\bm x)}(\theta \mid \bm x)\}_{m}} \max_{i\in 1,\ldots,n} | \alpha^{\ast}_{i} \times f_{\theta_m}(x_{i}) - \lambda_{i}\times f_{\theta_m}(x_{i})  | \\
&= \max_{\theta \in \{\xi^{\bm{\alpha}^{\ast}(\bm x)}(\theta \mid \bm x)\}_{m}} \max_{i\in 1,\ldots,n}|\left((1-\lambda_i)\alpha_i\right) \times f_{\theta_m}(x_{i}) | \leq \Delta_{\bm{\alpha},\mathbf{x}}\label{eq:reduce},
\end{align}

\noindent where we evaluate the log-pseudo likelihood over parameter draws, $\theta_{m}$, taken from the joint posterior distribution. We further reduce the risk-based weighted multiplier $\alpha_{i}$ by $(1-\lambda_{i})$ in Equation~(\ref{eq:reduce}). The $(1-\lambda_{i})$ term represents the portion of the data support that is not known by the interested public and, therefore, disclosure sensitive. Under our set-up, we account for intruder knowledge \emph{indirectly} through knowledge weights $\bm{\lambda}$ rather than directly truncating the data support.

\subsection{Range-averaged formal privacy}

Our definition of the probability $(1-\lambda_{i})$ that unit $i$ is included in the sensitive range of the data space leads us to deconstruct the likelihood into a sensitive and non-sensitive component.  We use the likelihood decomposition to define a new mechanism under an implied new privacy standard that is amplified relative to DP because it focuses solely on the sensitive portion of the likelihood contribution (by conditioning on known information about sensitive data ranges).  

We now make the definition of this new, amplified privacy standard explicit with a warm-up focus on the non-weighted posterior mechanism followed by a generalization to the risk-weighted pseudo posterior mechanism. Our new privacy standard conditions the posterior mechanism on the set of record-indexed ranges or balls chosen by the owner of the closely-held data.
 
Let data $(x_{1},\ldots,x_{n}) \in \mathcal{X}^{n}$, where each $x_{i}\in\mathbb{R}$. Suppose ranges, $R_{1}, …,R_{n} \in \mathcal{X}^{n} \subset \mathbb{R}^{n}$ are designated to be sensitive by the owner of the closely-held data.   The specification of this range gives rise to a probability that a datum will \emph{not} be contained in the sensitive range, $\lambda(X,R)$, defined for $X\in\mathcal{X}$, $R\in\mathbb{R}$ and produces likelihood weight, $\alpha^{\ast}(X,R)$. We then use the same definition of event probabilities as in DP and aDP, but now conditioning on $\mathbf{R}$.  

\begin{definition}
(Range-averaged Privacy under the Posterior Mechanism)
\begin{equation}
\mathop{\sup}_{\mathbf{x}\in \mathcal{X}^{n},\mathbf{x}^{'}\in \mathcal{X}^{n-1}:\delta(\mathbf{x}, \mathbf{x}^{'}) = 1} \mathop{\sup}_{B \in \beta_{\Theta}} \frac{\xi^{\bm{\lambda}^{c}(\mathbf{x})}(B \mid \mathbf{x},\mathbf{R})}{\xi^{\bm{\lambda}^{c}(\mathbf{x}^{'})}(B \mid \mathbf{x}^{'},\mathbf{R})} \leq e^{\epsilon}, \nonumber
\end{equation}
\end{definition}
\noindent where $\bm{\lambda}\left(\mathbf{x}\right) = \left(\lambda_{1}(\mathbf{x}),\ldots,\lambda_{n}(\mathbf{x})\right)$ for $\lambda_{i}$ defined as the probability that datum $x_{i}$ lies outside the range deemed sensitive by the owner of the data in Equation~(\ref{nonprivaterange}). Define $\delta(\mathbf{x},\mathbf{x}^{'})$ to be the number of records that differ between $\mathbf{x}$ and $
\mathbf{x}^{'}$ (with all other records being identical in both datasets).

We recall from \citet{SavitskyWilliamsHu2020ppm} that for a database sequence, $\mathbf{x} = \left(x_{1},\ldots,x_{n}\right) \in \mathcal{X}^{n}$  under $x_{1},\ldots,x_{n} \sim P_{\theta_{0}}$, for some $\theta_{0}\in\Theta$, we formulate the pseudo likelihood,
\begin{equation}
\label{pseudolike}
p_{\theta}^{\bm{\lambda}^{c}}(\mathbf{x}) = \mathop{\prod}_{i=1}^{n} p_{\theta_i}(x_i)^{1-\lambda_i(\mathbf{x})},
\end{equation}
for each $\theta \in \Theta$ and $\mathbf{x}\in\mathcal{X}^{n}$.  

\begin{equation}
\label{pseudopost}
\xi^{\bm{\lambda}^{c}(\mathbf{x})}(B \mid \mathbf{x}) = \frac{\int_{\theta \in B} p_{\theta}^{\bm{\lambda}^{c}}(\mathbf{x}) d \xi(\theta)(\mathbf{x}) }{\phi^{\bm{\lambda}^{c}(\mathbf{x})}},
\end{equation}
where $\phi^{\bm{\lambda}^{c}(\mathbf{x}) }(\mathbf{x}) \overset{\Delta}{=} \int_{\theta \in \Theta} p_{\theta}^{\bm{\lambda}^{c}}(\mathbf{x}) d \xi(\theta)$ normalizes the pseudo posterior distribution.

We do not marginalize or maximize over $\mathbf{R}$, but condition on it because it is known information.   $\lambda(X,R)$ marginalizes over data generating parameters $\theta \in \Theta$ that implicitly depends on model, $\mathcal{M}$.   We suppress that conditioning because the entire set-up implicitly depends on $\mathcal{M}$.

We note that the mechanism under the range-averaged privacy standard is simply the posterior distribution. Therefore, the appearance of $(1-\lambda_{i})$ reflects the portion of the likelihood contribution that is sensitive and needs to be protected.  We see this by examining the resulting sensitivity under range-averaged privacy,
\begin{align}
\Delta_{\bm{\lambda},\mathbf{x}} &= \max_{\theta \in \{\xi(\theta \mid \bm x)\}_{m}} \max_{i\in 1,\ldots,n} | f_{\theta_m}(x_{i}) - \lambda_{i}\times f_{\theta_m}(x_{i})  | \label{eq:subtract}\\
&= \max_{\theta \in \{\xi(\theta \mid \bm x)\}_{m}} \max_{i\in 1,\ldots,n}|\left(1-\lambda_i\right) \times f_{\theta_m}(x_{i}) | \leq \Delta_{\bm{\mathbf{x}}},
\end{align}
where the correct sensitivity, $\Delta_{\bm{\lambda},\mathbf{x}}$, is achieved by subtracting the non-sensitive portion of the data range in Equation~(\ref{eq:subtract}) from the full data log-likelihood.

\begin{definition}
(Range-averaged Privacy under the Pseudo Posterior Mechanism)
\begin{equation}
\mathop{\sup}_{\mathbf{x}\in \mathcal{X}^{n},\mathbf{x}'\in \mathcal{X}^{n-1}:\delta(\mathbf{x}, \mathbf{x}') = 1} \mathop{\sup}_{B \in \beta_{\Theta}} \frac{\xi^{\bm{\lambda}^c \bm{\alpha}(\mathbf{x})}(B \mid \mathbf{x},\mathbf{R})}{\xi^{\bm{\lambda}^c \bm{\alpha}(\mathbf{x}')}(B \mid \mathbf{x}',\mathbf{R})} \leq e^{\epsilon}, \nonumber
\end{equation}
\end{definition}

The mechanism associated with the range-restricted \emph{pseudo} posterior privacy standard is the pseudo posterior mechanism with the likelihood contribution for each unit $i \in (1,\ldots,n)$ exponentiated by $\alpha^{\ast}_{i}$.
 
This new range-averaged privacy standard can be viewed as relaxation of DP because it averages over datasets to compute $\lambda_{i}$ (rather than bounds a worst case), but it is also an amplification because it conditions on non-sensitive ranges. Yet, beyond comparing to DP, the range-averaged privacy represents a coherent, formal standard in its own right (e.g., like the Pufferfish formal privacy standard of \citet{e2c5145465454699a3b17fecbe979a05}).   Both the range-averaged and Pufferfish privacy standards provide amplifications of the privacy guarantee provided by DP by defining their standards to condition on additional information/assumptions. 

One may gain insight into the relative strength of the range-averaged privacy guarantee by focusing on the corner case of a highly risky record.  Under the usual aDP guarantee (with no range restriction) this datum's likelihood contribution would be assigned a low value for the risk weight, $\alpha$.  However, if the datum for this record would rarely fall within the sensitive range under the generating distribution for this record then the likelihood weight, $\alpha^{\ast}_i > \alpha_i$, with no loss of privacy protection since that otherwise risky record is deemed sensitive with a low probability, $(1-\lambda_i)$. 

We demonstrate the amplification in the privacy guarantee relative to aDP offered by the range-averaged privacy standard in simulation studies presented in Section~\ref{sec:simulation}.
We note that the degree of protection offered by range-averaged privacy may be increased relative to a baseline by simply selecting a wider range, $R_i$, for a datum that the data owner deems more sensitive, which will produce a lower $\lambda_i$.

We may now produce a similar formal privacy guarantee as \citet{SavitskyWilliamsHu2020ppm} that conditions on known sensitive ranges, $\mathbf{R}$.

\begin{theorem}
\label{th:dpresult}
$\forall \mathbf{x} \in \mathcal{X}^{n}, \mathbf{x}' \in \mathcal{X}^{n-1}:\delta(\mathbf{x}, \mathbf{x}') = 1, B \in \beta_{\Theta}$ (where $\beta_{\Theta}$ is the $\sigma-$algebra of measurable sets on $\Theta$) under $\bm{\alpha}(\cdot)$ with $\Delta_{\bm{\alpha},\bm{\lambda}, \mathbf{x}} > 0$, 
\begin{equation}
\mathop{\sup}_{B \in \beta_{\Theta}} \,\,\,  \frac{\xi^{\bm{\lambda}^c \bm{\alpha}(\mathbf{x})}(B \mid \mathbf{x})}{\xi^{\bm{\lambda}^c \bm{\alpha}(\mathbf{x}')}(B \mid \mathbf{x}')} \leq \exp(2\Delta_{\bm{\alpha},\bm{\lambda}, \mathbf{x}}),
\end{equation}
i.e., the pseudo posterior $\xi^{\bm{\lambda}^c \bm{\alpha}(\mathbf{x})}(\cdot \mid \mathbf{x})$ has local privacy guarantee $2\Delta_{\bm{\alpha},\bm{\lambda},\mathbf{x}}$.
\end{theorem}
\begin{proof}
    Replace $p_{\theta}^{\bm{\alpha}(\mathbf{x})}$ with $p_{\theta}^{\bm{\lambda}^{c}\bm{\alpha}(\mathbf{x})}$ and $\xi^{\bm{\alpha}(\mathbf{x})}\left(B\mid\mathbf{x}\right)$ with $\xi^{\bm{\lambda}^{c}\bm{\alpha}(\mathbf{x})}\left(B\mid\mathbf{x}\right)$ in proofs in Appendix of \citet{SavitskyWilliamsHu2020ppm}.
\end{proof}

We note that the guarantee of Theorem~\ref{th:dpresult} is \emph{local} to dataset $\mathbf{x} \in \mathcal{X}^{n}$.  The \emph{global} sensitivity, $\Delta_{\bm{\alpha},\bm{\lambda}} = \mathop{\sup_{\mathbf{x}\in \mathcal{X}^{n}}}\Delta_{\bm{\alpha},\bm{\lambda},\mathbf{x}}$ is constructed as the upper bound of the sensitivities for all datasets $\mathbf{x} \in \mathcal{X}^{n}$. As reviewed in Section~\ref{sec:intro:PPMreview}, \citet{SavitskyWilliamsHu2020ppm} demonstrate that for any realized $\mathbf{x}$ under a data generating process that $\Delta_{\bm{\alpha},\bm{\lambda},\mathbf{x}}$ contracts onto $\Delta_{\bm{\alpha},\bm{\lambda}}$ for sample size, $n$, sufficiently large (at $\mathcal{O}(n^{-\frac{1}{2}})$).  This relaxed privacy standard is hence labeled as aDP where a stands for asymptotic. We may extend the large sample aDP bound, $\epsilon \leq 2\Delta_{\bm{\alpha}}$, to the pseudo posterior mechanism under our amplified range-averaged privacy standard to achieve a new $\epsilon \leq 2\Delta_{\bm{\alpha},\bm{\lambda}}$ range-averaged aDP as the limiting point of $\Delta_{\bm{\alpha},\bm{\lambda},\mathbf{x}}$ given in Equation~(\ref{eq:reduce}).

\section{Truncation of record inclusion in sensitive range}\label{sec:truncated}
The range-averaged privacy standard introduced in Section \ref{sec:averaged} utilizes what may be viewed as the ``maximum" information about the sensitive range because it both uses the range and the frequency with which the data values for a record appear in that sensitive range over repeated sampling from the data generating distribution.  Even if the sensitive range itself is relatively wide, we are able to induce relatively less distortion to achieve a targeted privacy guarantee if the probability $(1-\lambda_i)$ that generates datum values for record $i$ will lie inside that range is small. The range-averaged privacy standard produces a probability that the generating distribution of each datum $i$ will lie inside the sensitive range using a model, $\mathcal{M}$, estimated on the closely-held data, rather than declaring each unit as lying in the sensitive range or not. Therefore, we allocate a $(1-\lambda_{i})$ proportion of the likelihood for record $i$ as sensitive.

Perhaps a more conservative approach would be to focus on using the ``minimum" information about the sensitive range.  We declare the datum for a record as either lying inside or outside of that range and condition our use of the model distribution only on endpoints of the sensitive interval, rather than the distribution of data values within the sensitive range as in the range-averaged standard.  

To motivate our development of this minimum use of the known sensitive range we return to the decomposition of the pseudo likelihood of the range-averaged pseudo posterior mechanism into the product of two components 
\begin{equation}\label{eq:lambdadecomp}
p_{\bm{\theta}}^{\bm{\alpha}{\ast}}(\mathbf{x}) = p_{\bm{\theta}}^{\bm{\lambda}^c \bm{\alpha}}(\mathbf{x}) p_{\bm{\theta}}^{\bm{\lambda}}(\mathbf{x})
= \mathop{\prod}_{i=1}^{n}p(x_{i}\mid\bm{\theta})^{(1-\lambda_i) \alpha_i}
\mathop{\prod}_{i=1}^{n}p(x_{i}\mid\bm{\theta})^{\lambda_{i}}.
\end{equation}
Since the second pseudo likelihood component on the right hand most term in Equation (\ref{eq:lambdadecomp}) is viewed as acceptable to expose to the public, we only use the first component to assess privacy risk.

We develop an alternative privacy standard by declaring a record's datum value to lie inside or outside of a sensitive data range with probability $1$, and then decompose the pseudo likelihood into two components based on interval censoring and its complement, interval truncation. 

For an arbitrary interval $[a\times x_i,b\times x_i] \in \mathcal{X}$, we propose an interval censored formulation for the assumed \emph{known} or non-sensitive likelihood component as
\[
p^I(x_i|\theta,a,b) = \left\{
\begin{array}{rl}
P(b\times x_i \mid \theta) - P(a\times x_i \mid \theta), & x_i \in [a\times x_i, b\times x_i] \\
p(x_i \mid \theta), & x_i \notin [a\times x_i,b\times x_i],
\end{array} \right.
\]
with $P(\cdot \mid \theta)$ being the cumulative density function (CDF) under our model for the closely-held data.  This likelihood represents information generally known because when $x_i$ is in the sensitive range and the endpoints of that range are known.  By contrast, when $x_i$ is not in the sensitive range there is no need to provide privacy protection. 

The sensitive complement (i.e., truncated) likelihood term is specified by

\[
p^{I^c}(x_i|\theta,a,b) = \left\{
\begin{array}{rl}
p(x_i|\theta)/\left( P(b\times x_i|\theta) - P(a\times x_i|\theta)\right), & x_i \in [a\times x_i,b\times x_i] \\
1, & x_i \notin [a\times x_i,b\times x_i].
\end{array} \right.
\]

We define both the sensitive and non-sensitive components of the likelihood by allocating datum $x_i$ dichotomously to the sensitive region.  If a record is assigned to the sensitive region of the likelihood complement, then we would fully apply our pseudo posterior risk-based weight, $\alpha_{i}$, to construct our pseudo posterior mechanism.  The known information is not conveyed through adjusting the privacy weight, $\alpha_{i}$, in the the likelihood as is done in the range-averaging standard, but rather dividing the likelihood, $p(x_i\mid\theta)$, by the difference in the CDF values at the endpoints, $\left( P(b\times x_i|\theta) - P(a\times x_i|\theta)\right)$.  Our resulting pseudo posterior mechanism is constructed as

\[
p^{I^c, \alpha_i}(x_i|\theta,a,b) = \left\{
\begin{array}{rl}
p(x_i|\theta)^{\alpha_i}/\left( P(b\times x_i|\theta) - P(a\times x_i|\theta)\right), & x_i \in [a\times x_i,b\times x_i ] \\
1, & x_i \notin [a\times x_i,b\times x_i]
\end{array} \right.
\]
All to say, our mechanism under the amplified privacy standard for the truncated sensitive range uses the same $\alpha-$weighted pseudo posterior as is used for the non-range-restricted aDP privacy standard. This use of the same $\alpha-$weighted pseudo posterior contrasts with the mechanism for the range-averaged standard that uses a revised $\alpha^{\ast}$ because unit datum values are probablistically allocated to sensitive ranges. We may summarize our range-truncated mechanism with,
\begin{equation}
p_{\bm{\theta}}^{\bm{\alpha}}(\mathbf{x}) = p_{\bm{\theta}}^{\bm{I}^c \bm{\alpha}}(\mathbf{x}) p_{\bm{\theta}}^{\bm{I}}(\mathbf{x})
= \mathop{\prod}_{i=1}^{n} p^{I^c, \alpha_i}(x_i \mid \theta,a,b)
\mathop{\prod}_{i=1}^{n}p^I(x_i \mid \theta,a,b).
\end{equation}

In practice, one will construct their sensitive range to be centered on or surrounding the datum value, $x_i \in [a\times x_i, b\times x_i]$ so that $x_i$ \emph{always} lies in the sensitive range. Then our sensitivity calculation focuses on the first complement:

\begin{align}
\Delta_{\bm{\alpha},\bm{I},\mathbf{x}}
&= \max_{\theta \in \{\xi^{\bm{\alpha}^{\ast}(\bm x)}(\theta \mid \bm x)\}_{m}} \max_{i\in 1,\ldots,n}\lvert\alpha_i \times f_{\theta_m}(x_{i}) -  \log(P(b\times x_i|\theta_m)- P(a\times x_i|\theta_m))\rvert.
\end{align}
Recall the non-sensitive portion of the likelihood is expressed under range-averaged privacy in Section \ref{sec:averaged} through exponentiating the datum likelihood contribution by $\lambda_{i}$. 
The larger is $0 \leq \lambda_{i} \leq 1$ the greater the portion of the datum likelihood admitted to the pseudo posterior model.  We may interpret probability $\lambda_i$ from range-average privacy as the effective sample size or fraction of information the data producer thinks is ``public" and thus does not need to be protected. 
By contrast, one divides the likelihood by the $P(b\times x_{i} \mid \theta) - P(a\times x_{i} \mid \theta)$ under the range-truncation privacy standard. For any finite $a$
and $b$, the difference in CDF functions is $< 1$, which magnifies the likelihood as compared to an unrestricted range (in which case the divisor is exactly $1$).   One may think of relatively risky records as located in the tails that expresses small likelihood values and the divisor would increase those values, which partially mitigates the $\bm{\alpha}$ down-weighting.  We expect $\Delta_{\bm{\alpha},\bm{I},\mathbf{x}} \le
 \Delta_{\bm{\alpha},\mathbf{x}}$. When the interval $[a\times x_i,b\times x_i]$ widens into the whole range, the second term is $\log(1)=0$, which causes the range-truncated privacy standard to revert to the aDP privacy standard. 

\subsection{Range-truncated formal privacy}

Our new privacy standard conditions the posterior mechanism on the set of record-indexed ranges or balls chosen by the owner of the closely-held data. Let data $(x_{1},\ldots,x_{n}) \in \mathcal{X}^{n}$, where each $x_{i}\in\mathbb{R}$. Suppose ranges, $R_{1}, …,R_{n} \in \mathcal{X}^{n} \subset \mathbb{R}^{n}$ are designated to be sensitive by the owner of the closely-held data.   Suppose also that each underlying sensitive value $x_{i} \in R_{i}$ by construction.
The specification of this range gives rise to the truncation adjustment $P_\theta(R_{i}) = \int_{x \in R_{i}} p_{\theta}(x) d x$ and the range-truncated likelihood $p_{\theta_i}(x_i)/P_\theta(R_{i})$. We then use our same definition of event probabilities in DP and aDP from \citet{SavitskyWilliamsHu2020ppm}, but now conditioning on known information, $\mathbf{R}$.


\begin{definition}
(Range-Truncated Privacy under the Posterior Mechanism)
\begin{equation}
\mathop{\sup}_{\mathbf{x}\in \mathcal{X}^{n},\mathbf{x}'\in \mathcal{X}^{n-1}:\delta(\mathbf{x}, \mathbf{x}') = 1} \mathop{\sup}_{B \in \beta_{\Theta}} \frac{\xi^{\bm{I}^{c}(\mathbf{x})}(B \mid \mathbf{x},\mathbf{R})}{\xi^{\bm{I}^{c}(\mathbf{x}')}(B \mid \mathbf{x}',\mathbf{R})} \leq e^{\epsilon}, \nonumber
\end{equation}
\end{definition}

For a database sequence, $\mathbf{x} = \left(x_{1},\ldots,x_{n}\right) \in \mathcal{X}^{n}$  under $x_{1},\ldots,x_{n} \sim P_{\theta_{0}}$, for some $\theta_{0}\in\Theta$, we formulate the truncated (pseudo-) likelihood,
\begin{equation}
\label{trunclike}
p_{\theta}^{\bm{I}^{c}}(\mathbf{x}) = \mathop{\prod}_{i=1}^{n} p_{\theta}(x_i)/P_{\theta}(R_i),
\end{equation}
for each $\theta \in \Theta$ and $\mathbf{x}\in\mathcal{X}^{n}$.  

\begin{equation}
\label{pseudopost}
\xi^{\bm{I}^{c}(\mathbf{x})}(B \mid \mathbf{x}) = \frac{\int_{\theta \in B} p_{\theta}^{\bm{I}^{c}}(\mathbf{x}) d \xi(\theta)(\mathbf{x}) }{\phi^{\bm{I}^{c}(\mathbf{x})}},
\end{equation}
where $\phi^{\bm{I}^{c}(\mathbf{x}) }(\mathbf{x}) \overset{\Delta}{=} \int_{\theta \in \Theta} p_{\theta}^{\bm{I}^{c}}(\mathbf{x}) d \xi(\theta)$ normalizes the pseudo posterior distribution.
We note that the truncated interval offset $P_{\theta}(R_i)$ explicitly depends on the model and thus we include the $\theta$ subscript.

Now, the mechanism under the range-truncated privacy standard is simply the posterior distribution. Therefore, the appearance of $P_{\theta}(R_i)$ reflects the portion of the likelihood contribution that is not sensitive and can be used without ``penalty".  We see this by examining the resulting sensitivity under range restricted privacy,
\begin{align}
\Delta_{\bm{I},\mathbf{x}} &= \max_{\theta \in \{\xi(\theta \mid \bm x)\}_{m}} \max_{i\in 1,\ldots,n} | f_{\theta_m}(x_{i}) -  f_{\theta_m}(R_{i}) |,
\end{align}
where $f_{\theta_m}(R_{i}) = \log(P_{\theta_m}(R_i))$.
We expect that $\Delta_{\bm{I},\mathbf{x}} \le \Delta_{\mathbf{x}}$, which provides an amplification of aDP. While we can not guarantee that $(f_{\theta_m}(x_{i}) -  f_{\theta_m}(R_{i})) < 0$ for every $i$ there may be some value of $\theta$, $x_i$ with a very narrow $R_i$ that produce positive values. This would only occur for continuous outcomes with probability densities. Discrete outcomes have probability mass functions and thus would never have this issue.
However, we expect the largest magnitude $|f_{\theta_m}(x_{i}) -  f_{\theta_m}(R_{i})|$ to occur in the tails for very small values of $p_{\theta}(x_i)$ and thus very large negative values for $f_{\theta_m}(x_{i})$. 

\begin{definition}
(Range-Truncated Privacy under the Pseudo Posterior Mechanism)
\begin{equation}
\mathop{\sup}_{\mathbf{x}\in \mathcal{X}^{n},\mathbf{x}'\in \mathcal{X}^{n-1}:\delta(\mathbf{x}, \mathbf{x}') = 1} \mathop{\sup}_{B \in \beta_{\Theta}} \frac{\xi^{\bm{I}^c \bm{\alpha}(\mathbf{x})}(B \mid \mathbf{x},\mathbf{R})}{\xi^{\bm{I}^c \bm{\alpha}(\mathbf{x}')}(B \mid \mathbf{x}',\mathbf{R})} \leq e^{\epsilon}, \nonumber
\end{equation}
\end{definition}

The mechanism associated with the range-restricted pseudo posterior privacy standard is the pseudo posterior mechanism with the likelihood contribution for each unit $i \in (1,\ldots,n)$ exponentiated by $\alpha_{i}$.

\begin{theorem}
\label{th:dpresult}
$\forall \mathbf{x} \in \mathcal{X}^{n}, \mathbf{x}' \in \mathcal{X}^{n-1}:\delta(\mathbf{x}, \mathbf{x}') = 1, B \in \beta_{\Theta}$ (where $\beta_{\Theta}$ is the $\sigma-$algebra of measurable sets on $\Theta$) under $\bm{\alpha}(\cdot)$ with $\Delta_{\bm{\alpha},\bm{I}, \mathbf{x}} > 0$, 
\begin{equation}
\mathop{\sup}_{B \in \beta_{\Theta}} \,\,\,  \frac{\xi^{\bm{I}^c \bm{\alpha}(\mathbf{x})}(B \mid \mathbf{x})}{\xi^{\bm{I}^c \bm{\alpha}(\mathbf{x}')}(B \mid \mathbf{x}')} \leq \exp(2\Delta_{\bm{\alpha},\bm{I}, \mathbf{x}}),
\end{equation}
i.e., the pseudo posterior $\xi^{\bm{I}^c \bm{\alpha}(\mathbf{x})}(\cdot \mid \mathbf{x})$ has local privacy guarantee $2\Delta_{\bm{\alpha},\bm{I},\mathbf{x}}$.
\end{theorem}
\begin{proof}
    Replace $p_{\theta}^{\bm{\alpha}(\mathbf{x})}$ with $p_{\theta}^{\bm{I}^{c}\bm{\alpha}(\mathbf{x})}$ and $\xi^{\bm{\alpha}(\mathbf{x})}\left(B\mid\mathbf{x}\right)$ with $\xi^{\bm{I}^{c}\bm{\alpha}(\mathbf{x})}\left(B\mid\mathbf{x}\right)$ in proofs in Appendix of \citet{SavitskyWilliamsHu2020ppm}.
\end{proof}

Range-truncated aDP may be viewed as an amplification of aDP without any relaxation since the observed datum is fully allocated to the sensitive range, rather than using averaging to allocate some portion of its contribution as in range-averaged aDP introduced in Section \ref{sec:averaged}.

\subsection{General formulation for range-restricted privacy}
We may more fully understand our range-averaged and range-truncated privacy standards by viewing them as special cases of a more general privacy framework that incorporates sensitive range information to sharpen the privacy guarantee.

Let $\bm{\gamma}(\mathbf{R},\bm{\theta}, \mathcal{M})$ be index parameters that represent the contribution for assumed publicly known ranges, $\mathbf{R} = (R_{1},\ldots,R_{n}) \in \mathbb{R}^{n}$.

Define the pseudo likelihood decomposition between known or insensitive contributions, on the one hand, and sensitive contributions, on the other hand.

\begin{equation}
p_{\bm{\theta}}^{\bm{\alpha}^{\ast}}(\mathbf{x}) = p_{\bm{\theta},\bm{\gamma}^c}^{\bm{\alpha}}(\mathbf{x}) \times p_{\bm{\theta},\bm{\gamma}}(\mathbf{x}),
\end{equation}
where $\bm{\alpha}^{\ast}$ denotes the range-restricted, risk-based weights for the pseudo posterior mechanism while $\bm{\alpha}$ denotes the original risk-based weights for the pseudo posterior mechanism with respect to the aDP standard. We suppress the dependence of $\bm{\gamma}$ on $\mathbf{R}$ for readability.

Under the range-truncated standard,
$\gamma_{i}(\theta) = P(R_i\mid\theta)$ with $p_{\theta,\gamma_{i}^c}^{\bm{\alpha}} = p_{i}^{\alpha}/\gamma_{i}(\theta)$ and $p_{\theta,\gamma_{i}} = \gamma_{i}(\theta)$ so that $\alpha^{\ast}_i = \alpha_i$.  We recall that the sensitive range is defined under the range-truncated standard such that $x_{i}$ is always in the sensitive range.  

By contrast, $\gamma_{i}(\mathcal{M}) = \lambda_{i} = \Pr_{\mathcal{M},R_i}\left(x^{\ast}_{i} \notin R_i \right)$ under range-averaged privacy with $p_{\theta,\bm{\gamma}^c}^{\bm{\alpha}} = p_{\theta}^{(1-\gamma_{i})\alpha}$ and $p_{\theta,\gamma_{i}}= p_{\theta}^{\gamma_{i}}$.  This produces $\alpha^{\ast}_i = (1-\gamma_{i})\alpha_i + \gamma_{i}$.

The general formulation for the range-restricted model sensitivity, $\Delta_{\bm{\alpha},\bm{\gamma},\mathbf{x}} $, is solely based on the sensitive portion of the log-pseudo likelihood,

\begin{align}
\Delta_{\bm{\alpha},\bm{\gamma},\mathbf{x}}
&= \max_{\theta \in \{\xi^{\bm{\alpha}^{\ast}(\bm x)}(\theta \mid \bm x)\}_{m}} \max_{i\in 1,\ldots,n}\lvert f_{\bm{\theta}_m,\gamma_{i,m}^{c}}^{\alpha_{i}}(x_{i}) \rvert,
\end{align}

\noindent where $f_{\bm{\theta}_m,\gamma_{i,m}^{c}}^{\alpha_{i}}(x_{i}) = \log\left[p_{\bm{\theta}_{m},\gamma_{i,m}^c}^{\alpha_{i}}(x_{i})\right]$.  Focusing on the sensitive likelihood complement provides the privacy amplification under both the range-averaged and range-truncated privacy standards.  Examining the application to range-averaged privacy of the sensitivity construction from the likelihood complement,
$\Delta_{\bm{\alpha},\bm{\gamma},\mathbf{x}} = \max_{\theta \in \{\xi^{\bm{\alpha}^{\ast}(\bm x)}(\theta \mid \bm x)\}_{m}} \max_{i\in 1,\ldots,n} \lvert (1-\gamma_{i})\times \alpha_{i} \times f_{\theta_m}(x_{i}) \rvert$, where
$\bm{\gamma} = \bm{\lambda}$ and $\gamma_{i}^{c} = 1-\gamma_{i}$ under the range-averaged standard.  For the range-truncated standard, $\gamma_{i}(\theta) = P(b\times x_{i}\mid \theta) - P(a\times x_{i}\mid \theta)$ such that $\gamma_{i}^{c}(\theta) = 1/\gamma_{i}(\theta)$. Finally, $f_{\bm{\theta}_m,\gamma_{i,m}^{c}}^{\alpha_{i}}(x_{i}) = \alpha_{i} \times f_{\bm{\theta}_{m}}(x_{i}) - \log\gamma_{i}(\theta_{m})$.

We note that for the range-truncated standard $\gamma_{i}(\theta)$ depends only on the current models parameters and not on the underlying generating (unweighted) model $\mathcal{M}$. In contrast, for the range-averaged standard, $\gamma_{i}(\mathcal{M})$ depends on the predictive distribution of the generating model $\mathcal{M}$ but not explicitly on its parameters (which are integrated out) nor on the current model parameters $\theta$. 


We now proceed to a series of extensive simulation studies to illustrate the privacy amplication effects of the two novel range-restricted aDP standards.

\section{Simulation studies}
\label{sec:simulation}

In this section, we first illustrate how privacy guarantee strengthens the two novel range-restricted standards compared to the aDP standard in Section \ref{sec:simulation:story1}. Section \ref{sec:simulation:story2} is devoted to illustrate how the range-restricted standards can achieve higher utility for the same privacy budget compared to the aDP standard. Finally in Section \ref{sec:simulation:story3}, we focus on assigning a wider sensitive range for records in the right tail of the confidential data distribution to provide a higher level of privacy guarantee. 

\subsection{Privacy guarantee strengthens under restricted sensitive range}
\label{sec:simulation:story1}

Over repeated simulations, we generate $n = 2000$ records from $z \sim \textrm{Normal}(2, 1)$ and $x \sim \textrm{Lognormal}(z+1, 1)$ in order to mimic highly skewed data in real applications. For both range-averaged and range-truncated standards, we use $(a, b) = \{(0.4, 1.8), (0.6, 1.2)\}$ for two sets of sensitive bounds to represent the sensitive range information for which the data disseminator attempts to provide privacy protection. We use $S = 1000$ number of values generated to calculate $\lambda_i$ for the range-averaged standard.

Figure \ref{fig:sim-Lbounds-single} shows a single simulation example of the by-record Lipchitz values of the Unweighted synthesizer, the Weighted synthesizer (labeled as (-Inf, Inf) to represent no bounds of $(a, b)$), the two range-averaged synthesizers (labeled as (0.4, 1.8) avg and (0.6, 1.2) avg), and the two range-truncated synthesizers (labeled as (0.4, 1.8) trunc and (0.6, 1.2) trunc)). The privacy budget, $\epsilon_{\mathbf{x}}$, for each synthesizer is calculated as twice of the maximum Lipschitz bound. We can see that the Weighted synthesizer has the highest privacy budget. The two range-truncated synthesizers, namely the (0.4, 1.8) trunc and the (0.6, 1.2) trunc, are second and third. The two range-averaged synthesizers, namely the (0.4, 1.8) avg and (0.6, 1.2) avg, have the smallest privacy budgets. We note that the lower the privacy budget, the stronger the privacy guarantee.

As evident in Figure \ref{fig:sim-Lbounds-single}, for a single sample, the privacy budget decreases under a finite sensitive range for both range-averaged  and range-truncated  synthesizers, compared to the Weighted synthesizer. In other words, the two sets of range-restricted synthesizers produce stronger privacy guarantees compared to the Weighted. This is expected since the range-averaged and the range-truncated synthesizers utilize the information from the data disseminator and provide privacy protection to the sensitive range, whereas the Weighted synthesizer attempts to protect the full range. Moreover, for either range-restricted standard, shortening the sensitive bounds $(a, b)$ from (0.4, 1.8) to (0.6, 1.2) further decreases the privacy budget, corresponding to privacy protection for a narrower sensitive range. The decreasing effect is more dramatic for the range-averaged synthesizer, as can be seen by the sharp decrease of the maximum Lipchitz bound from (0.4, 1.8) avg to (0.6, 1.2) avg, compared to the range-truncated pair, namely (0.4, 1.8) trunc and (0.6, 1.2) trunc. 

Within same sensitive bounds $(a, b)$ choice (e.g. the pair of blue for $(0.4, 1.8)$ and the pair of green for $(0.6, 1.2)$), the range-averaged synthesizer produces further privacy budget decrease (i.e., stronger privacy guarantee) compared to the range-truncated synthesizer. We recall that the two sets of range-restricted synthesizers are designed to use information from the sensitive range in different ways. 
The range-averaged synthesizer, on the one hand, uses distributional information by computing probabilities of falling within a sensitive range. The range-truncated synthesizer, on the other hand, uses only the end points of the sensitive range. 
In other words, the range-averaged synthesizer uses more ``public" information than the range-truncated synthesizer, and hence produces a stronger, more focused and amplified privacy guarantee.  Yet, the range-averaged privacy standard defines and incorporates $1-\lambda_{i}$, a probability that record $i$ is sensitive, by averaging over draws from the model posterior predictive distribution.  While the privacy standard is mathematically ``formal" in that it is provable and independent of putative intruder behaviors, the averaging is different from the pure upper bounding that characterizes differential privacy.  By contrast, the range-truncated uses less information focused solely on the end points of the sensitive range and avoids any averaging.

\begin{figure}[H]
  \centering
   \includegraphics[width=1\textwidth]{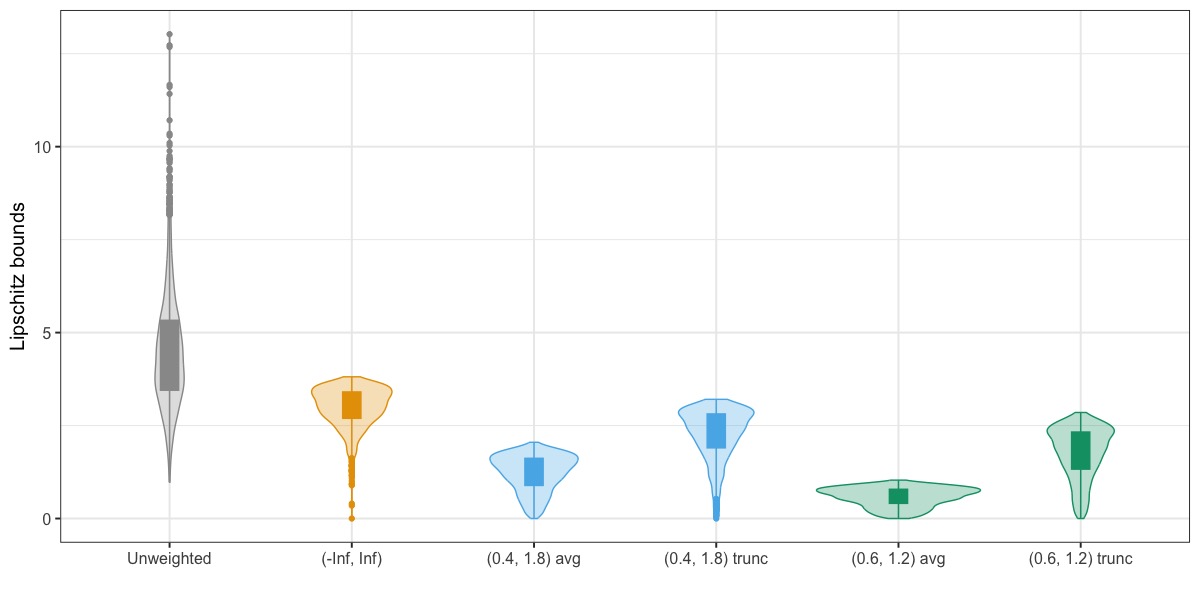}
      \caption{Violin plots of by-record Lipschitz bounds of Unweighted, (-Inf, Inf) (i.e., no bounds as Weighted), $(0.4, 1.8)$ averaged, $(0.4, 1.8)$ truncated, $(0.6, 1.2)$ averaged, and $(0.6, 1.2)$ truncated, over a single sample.}
      \label{fig:sim-Lbounds-single}
\end{figure}

Figure \ref{fig:sim-Lbounds-repeated} shows results from a Monte Carlo simulation of 100 repeated samples of the same data generating model. Once again, we observe lower privacy budgets, i.e., stronger privacy guarantees, from the two sets of range-restricted synthesizers compared to the Weighted synthesizer. Moreover, shortening the sensitive range $(a, b)$ provides a stronger privacy guarantee for either range-restricted synthesizer. Lastly, the range-averaged synthesizer produces a stronger and amplified privacy guarantee than the range-truncated synthesizer with the same $(a, b)$ choice. 

\begin{figure}[H]
  \centering
   \includegraphics[width=1\textwidth]{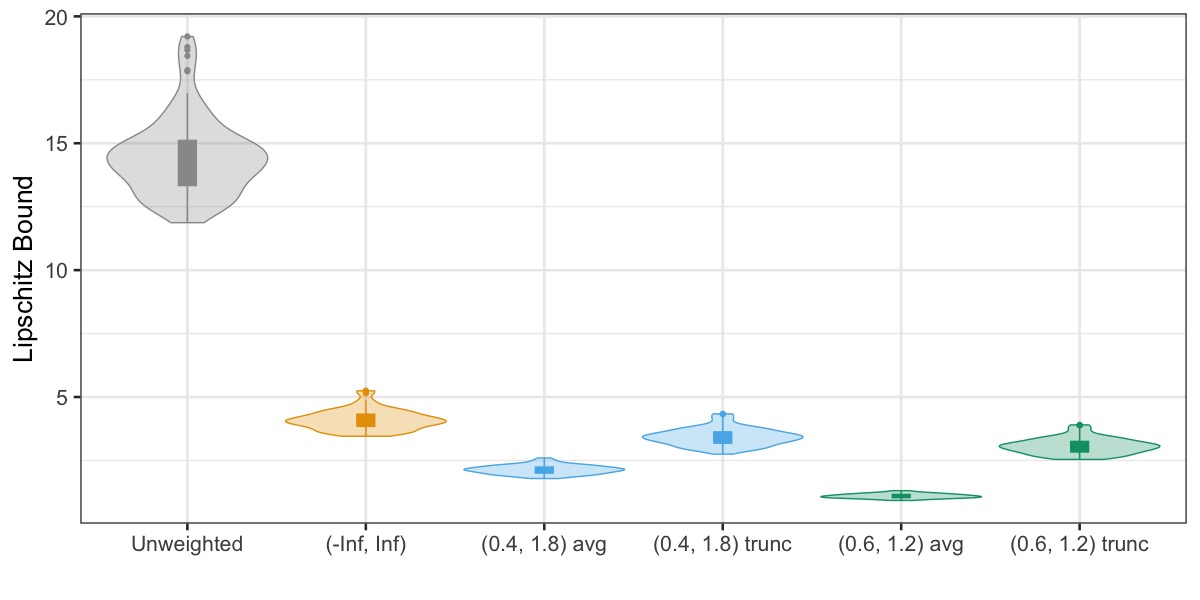}
      \caption{Violin plots of Lipschitz bounds of Unweighted, (-Inf, Inf) (i.e., no bounds as Weighted), $(0.4, 1.8)$ averaged, $(0.4, 1.8)$ truncated, $(0.6, 1.2)$ averaged, and $(0.6, 1.2)$ truncated, over 100 repeated samples.}
      \label{fig:sim-Lbounds-repeated}
\end{figure}

\begin{figure}[H]
  \centering
   \includegraphics[width=1\textwidth]{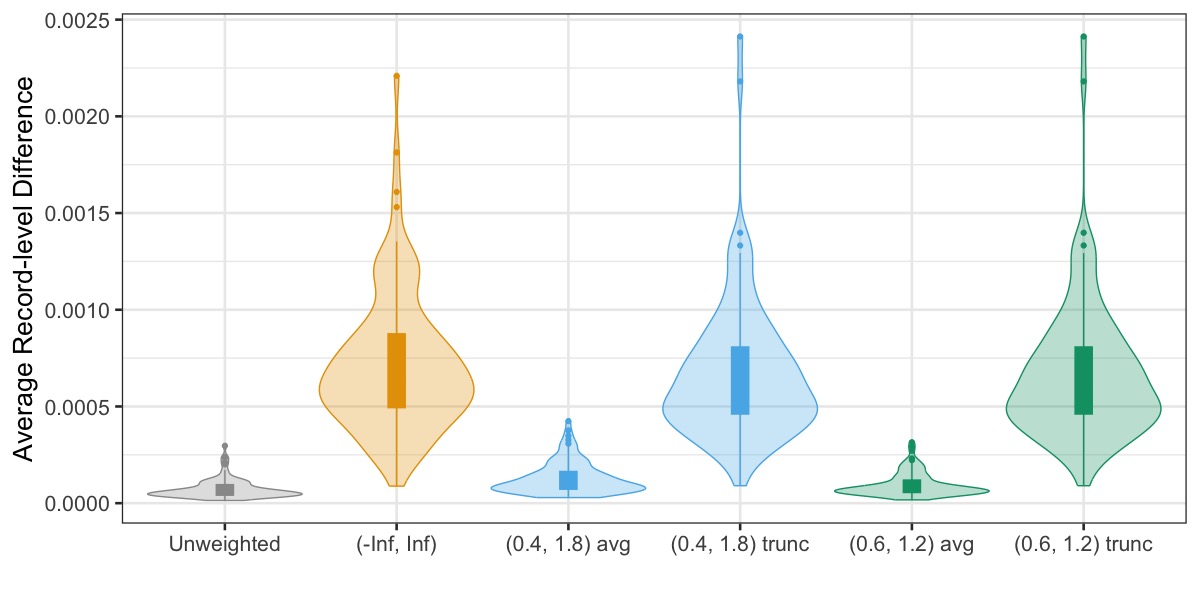}
      \caption{Violin plots of average metric of ECDF of Unweighted, (-Inf, Inf) (i.e., no bounds as Weighted), $(0.4, 1.8)$ averaged, $(0.4, 1.8)$ truncated, $(0.6, 1.2)$ averaged, and $(0.6, 1.2)$ truncated, over 100 repeated samples.}
      \label{fig:sim-ua}
\end{figure}

\begin{figure}[H]
  \centering
   \includegraphics[width=1\textwidth]{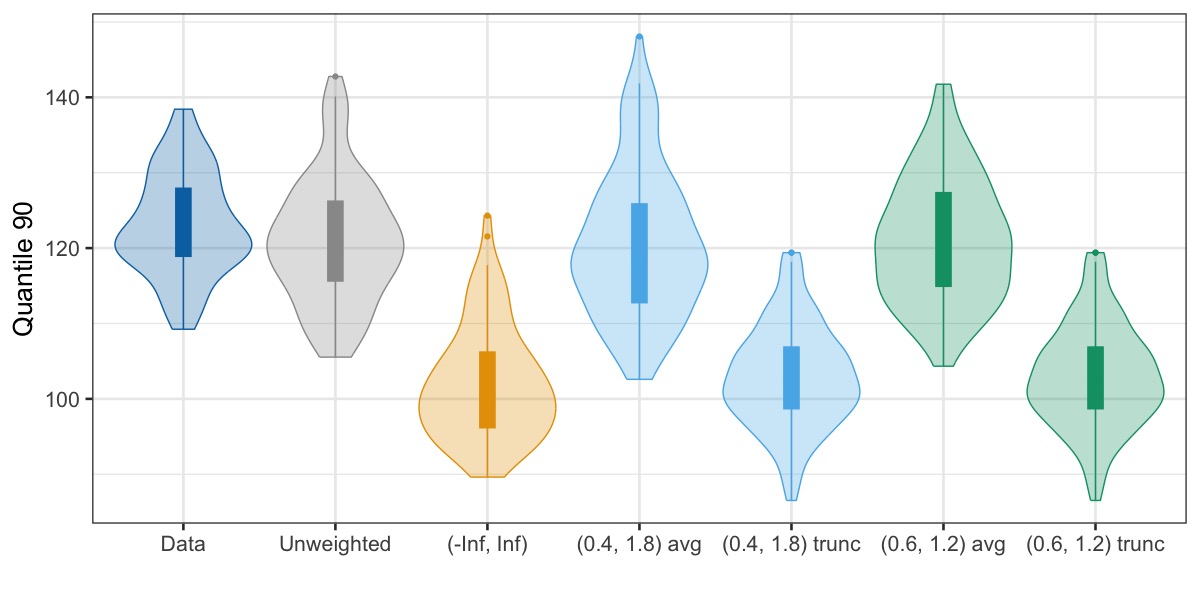}
      \caption{Violin plots of Q90s of Unweighted, (-Inf, Inf) (i.e., no bounds as Weighted), $(0.4, 1.8)$ averaged, $(0.4, 1.8)$ truncated, $(0.6, 1.2)$ averaged, and $(0.6, 1.2)$ truncated, over 100 repeated samples.}
      \label{fig:sim-Q90s}
\end{figure}

We now consider the utility of the simulated synthetic data from this set of synthesizers. These results are based on 100 repeated samples. For global utility that focuses on distributional similarities of the confidential and the synthetic data, we present the average record-level difference from the ECDF statistics in Figure \ref{fig:sim-ua} \citep{Woo2009JPC, HuBowen2024WIRES}. For this utility metric, the smaller the value, the higher the utility. We can see that the range-averaged synthesizers produce synthetic data with the highest utility. Thus it accomplishes both a privacy amplification and improvement in utility of the synthetic data.  By contrast, the utility is the same for the Weighted and range-truncated synthesizers (up to Monte Carlo sampling error) because the \emph{same} privacy $\bm{\alpha}$ risk-based weights are used under both.  Then for the same utility, we see that range-truncated produces a stronger or amplified privacy guarantee than does the Weighted synthesizer. The shortening of the sensitive range $(a, b)$ further improves the utility for either range-restricted synthesizer. 

For analysis-specific utility, we focus on the 90\% quantile statistic in Figure \ref{fig:sim-Q90s}. For this utility metric, the closer to the data distribution (far left), the higher the utility of a synthesizer. As with the global utility metric, the range-averaged synthesizers preserve the highest utility. Similarly, a smaller sensitive range $(a, b)$ corresponds to higher utility. 

In summary, both sets of range-restricted synthesizers produce a stronger and amplified privacy guarantee than the Weighted by utilizing the information possessed by the data disseminator. By using more of the information through computing probabilities of being in the sensitive range, the range-averaged synthesizers expends less privacy budget than the range-truncated synthesizers for the same sensitive range choice of $(a, b)$. Moreover, the amplified privacy guarantee does not come at the price of compromised utility preservation. In fact, the range-averaged synthesizers produce synthetic data with very high utility. Lastly, shortening the sensitive range $(a, b)$ further enhances the privacy guarantee (i.e., reduces the privacy budget and achieves a stronger privacy guarantee) and improves utility, most evidently for the range-averaged synthesizers.

Lastly, Figure \ref{fig:sim-Lbounds-ns} confirms that incorporating a sensitive range, either through the range-averaged standard or the range-truncated standard, maintains the asymptotic feature of aDP that as the sample size increases, the local $\Delta_{\bm{\alpha},\mathbf{x}}$ contracts onto the global $\Delta_{\bm{\alpha}}$.

\begin{figure}[H]
  \centering
   \includegraphics[width=1\textwidth]{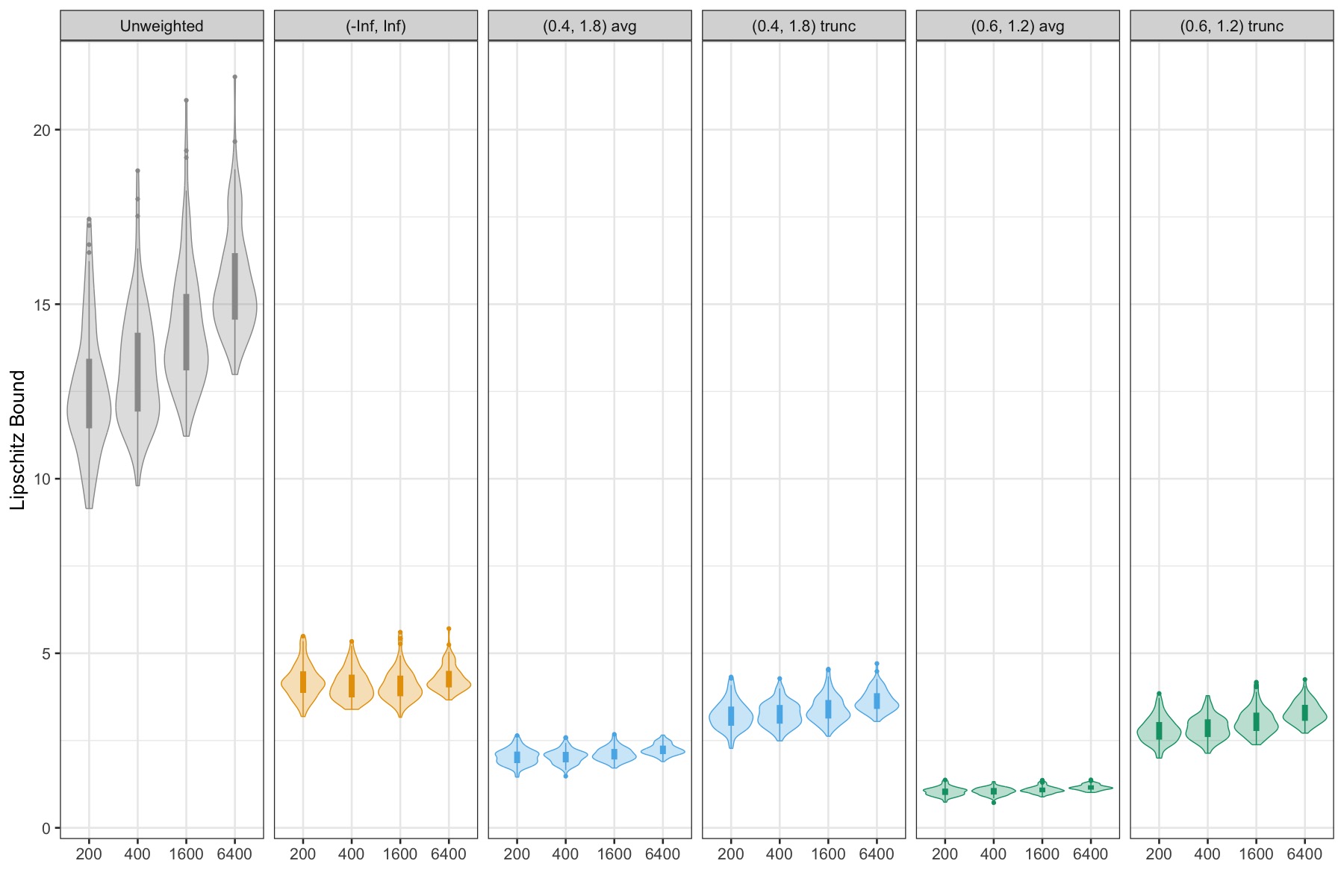}
      \caption{Violin plots of Lipschitz bounds of Unweighted, (-Inf, Inf) (i.e., no bounds as Weighted), $(0.4, 1.8)$ averaged, $(0.4, 1.8)$ truncated, $(0.6, 1.2)$ averaged, and $(0.6, 1.2)$ truncated, over 100 repeated samples and four different sample sizes $n = \{200, 400, 1600, 6400\}$.}
      \label{fig:sim-Lbounds-ns}
\end{figure}

\subsection{Restricting sensitive ranges achieves higher utility for the same privacy budget}
\label{sec:simulation:story2}

Data disseminators may wish to compare data synthesized using sensitive range information to data synthesized without using that information with both methods calibrated to the same privacy budget. Our newly proposed range-restricted standards offer such flexibility and in this part of the simulation, we focus on demonstrating the flexibility through the range-truncated standard. 

We continue with our data generating model of $z \sim \textrm{Normal}(2, 1)$ and $x \sim \textrm{Lognormal}(z+1, 1)$ for $n = 2000$ records for a single sample. In order to reach similar privacy budgets across different synthesizers, we scale down the privacy weights $\alpha_i$'s of the Weighted and those of the (0.4, 1.8) trunc synthesizers to bring down their privacy budgets to close to that of the (0.6, 1.2) trunc synthesizer. Such scaling achieves a similar privacy budget, roughly $\epsilon_{\mathbf{x}} = 5.7$, across the three synthesizers. 
The scaling constants are 0.73 for the Weighted and 0.88 for the (0.4, 1.8) trunc synthesizer, where the constant is 1 for the (0.6, 1.2) trunc synthesizer (i.e., no scaling). 

Figure \ref{fig:sim-Lbounds-single-truncated-similar-eps} presents the by-record Lipcthiz values of the values of the Unweighted synthesizer, the
Weighted synthesizer (labeled as (-Inf, Inf) to represent no bounds of (a, b)), and the two range-truncated synthesizers
(labeled as (0.4, 1.8) trunc and (0.6, 1.2) trunc)). The privacy budget, $\epsilon_{\mathbf{x}}$, for each synthesizer is calculated as twice of the maximum Lipschitz bound. We can see that through appropriate amount of scaling of the privacy weights $\alpha_i$'s, the three synthesizers now have similar privacy budgets, roughly $\epsilon_{\mathbf{x}} = 5.7$. The scaling effects of the privacy weights $\alpha_i$'s in each synthesizer are manifested in Figure \ref{fig:sim-alphas-single-truncated-similar-eps}, where the Unweighted synthesizer have all weights at 1 by design, and the (0.6, 1.2) trunc synthesizer has the highest weights, followed by the (0.4, 1.8) trunc synthesizer and the Weighted synthesizer. 

\begin{figure}[H]
  \centering
   \includegraphics[width=1\textwidth]{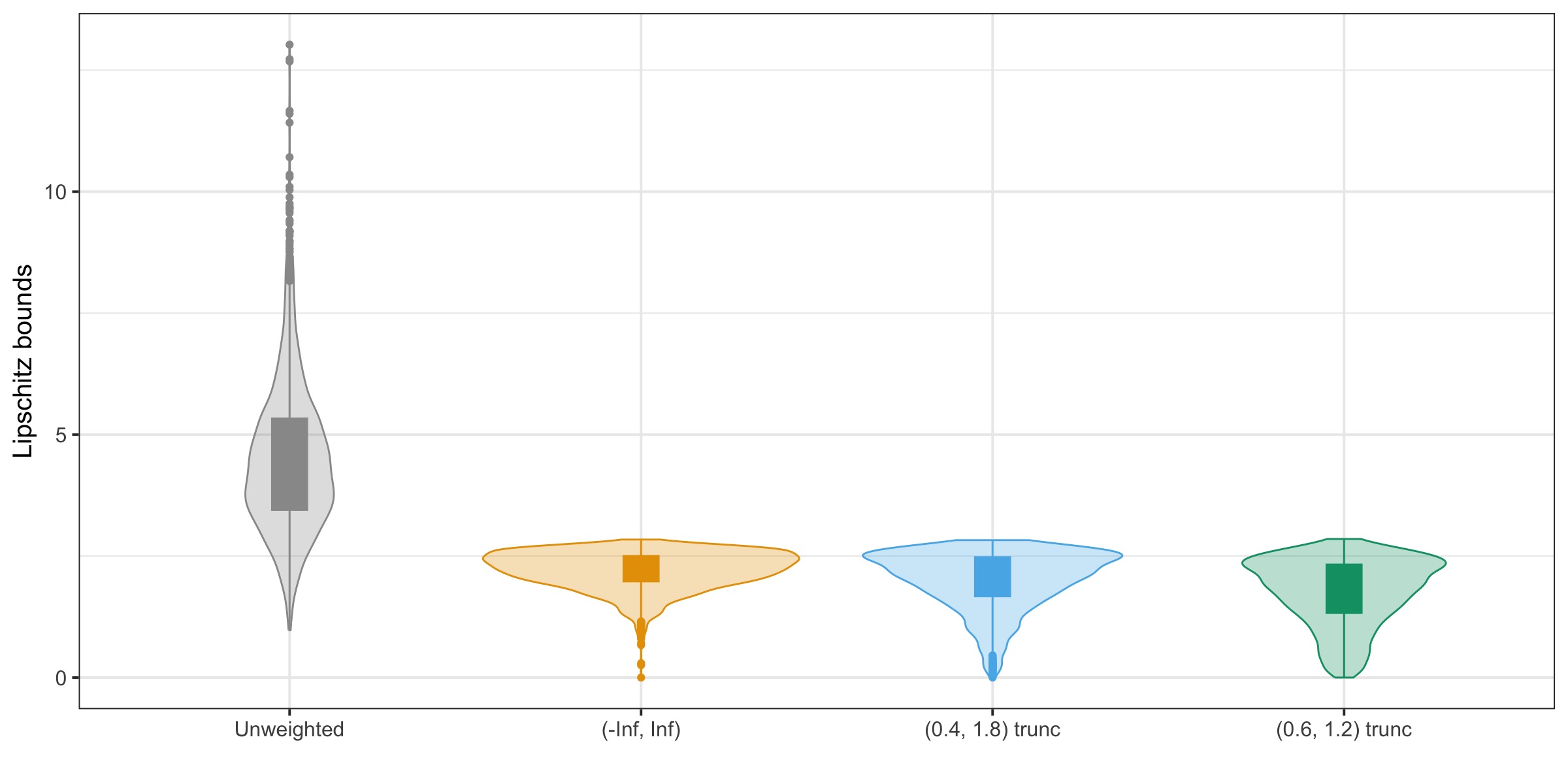}
      \caption{Violin plots of Lipschitz bounds of Unweighted, (-Inf, Inf) (i.e., no bounds as Weighted), $(0.4, 1.8)$ truncated, and $(0.6, 1.2)$ truncated, over a singe sample.}
      \label{fig:sim-Lbounds-single-truncated-similar-eps}
\end{figure}

\begin{figure}[H]
  \centering
   \includegraphics[width=1\textwidth]{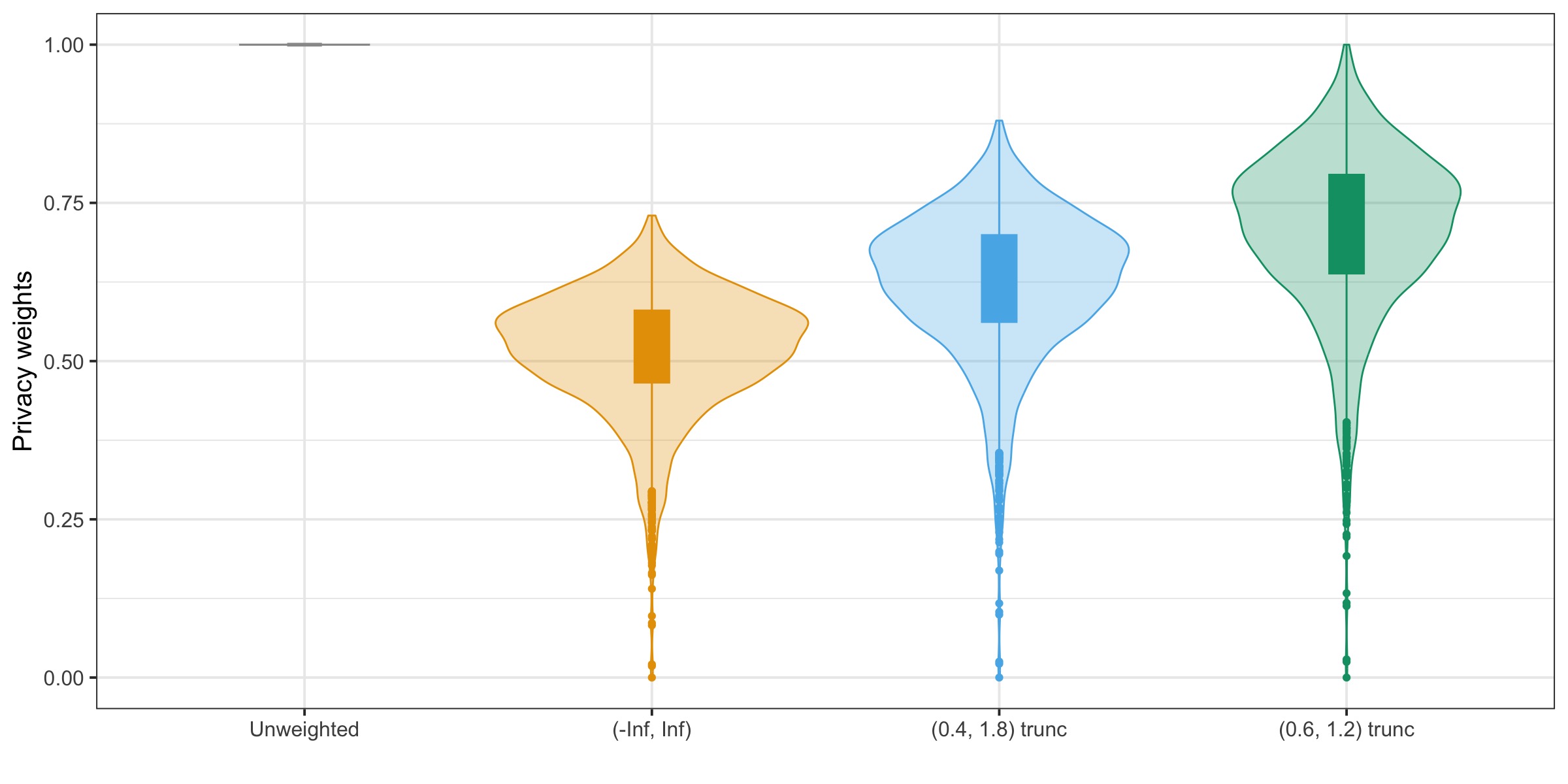}
      \caption{Violin plots of privacy weights of Unweighted, (-Inf, Inf) (i.e., no bounds as Weighted), $(0.4, 1.8)$ truncated, and $(0.6, 1.2)$ truncated, over a singe sample.}
    \label{fig:sim-alphas-single-truncated-similar-eps}
\end{figure}

After achieving similar privacy bugdets through scaling of the weights $\alpha_i$'s, we now turn to comparing the utility results, where we expect to achieve higher utility under sensitive range. Recall that results in Section \ref{sec:simulation:story1}, specifically Figure \ref{fig:sim-ua} and Figure \ref{fig:sim-Q90s}, suggest that without the scaling of the privacy weights $\alpha_i$'s, the Weighted synthesizer produces synthetic data with the same utility as range-truncated. Now, by scaling down the privacy weights of the Weighted synthesizer and the (0.4, 1.8) trunc synthesizer, we expect to see a utility reduction of these two synthesizers due to the decrease in privacy weights \citep{HuSavitskyWilliams2022JSSAM, HuWilliamsSavitsky2025SS}. 

Figures \ref{fig:sim-mean-single-truncated-similar-eps} through \ref{fig:sim-Q90-single-truncated-similar-eps} present analysis-specific utility of the means, medians, and 90\% quantiles. Across the three synthesizers, the (0.6, 1.2) trunc synthesizer has the highest utility in all three analysis-specific utility metrics across the three synthesizers, followed by the (0.4, 1.8) trunc synthesizer and then the Weighted synthesizer.

These results confirm that by achieving similar privacy budgets through down-scaling of the privacy weights $\alpha_i$'s of the Weighted and the (0.4, 1.8) trunc synthesizers, we produce synthetic data with the highest utility with the (0.6, 1.2) trunc synthesizer. In fact, both the (0.4, 1.8) trunc and the (0.6, 1.2) trunc synthesizers have higher utility than the Weighted, demonstrating the flexibility of our newly proposed range-truncated standards in producing synthetic data with higher utility and a similar privacy guarantee. Such flexibility widens the range of tools of data disseminators when it comes to incorporating sensitive information while fine tuning the balance of utility and privacy. 

\begin{figure}[H]
  \centering
   \includegraphics[width=1\textwidth]{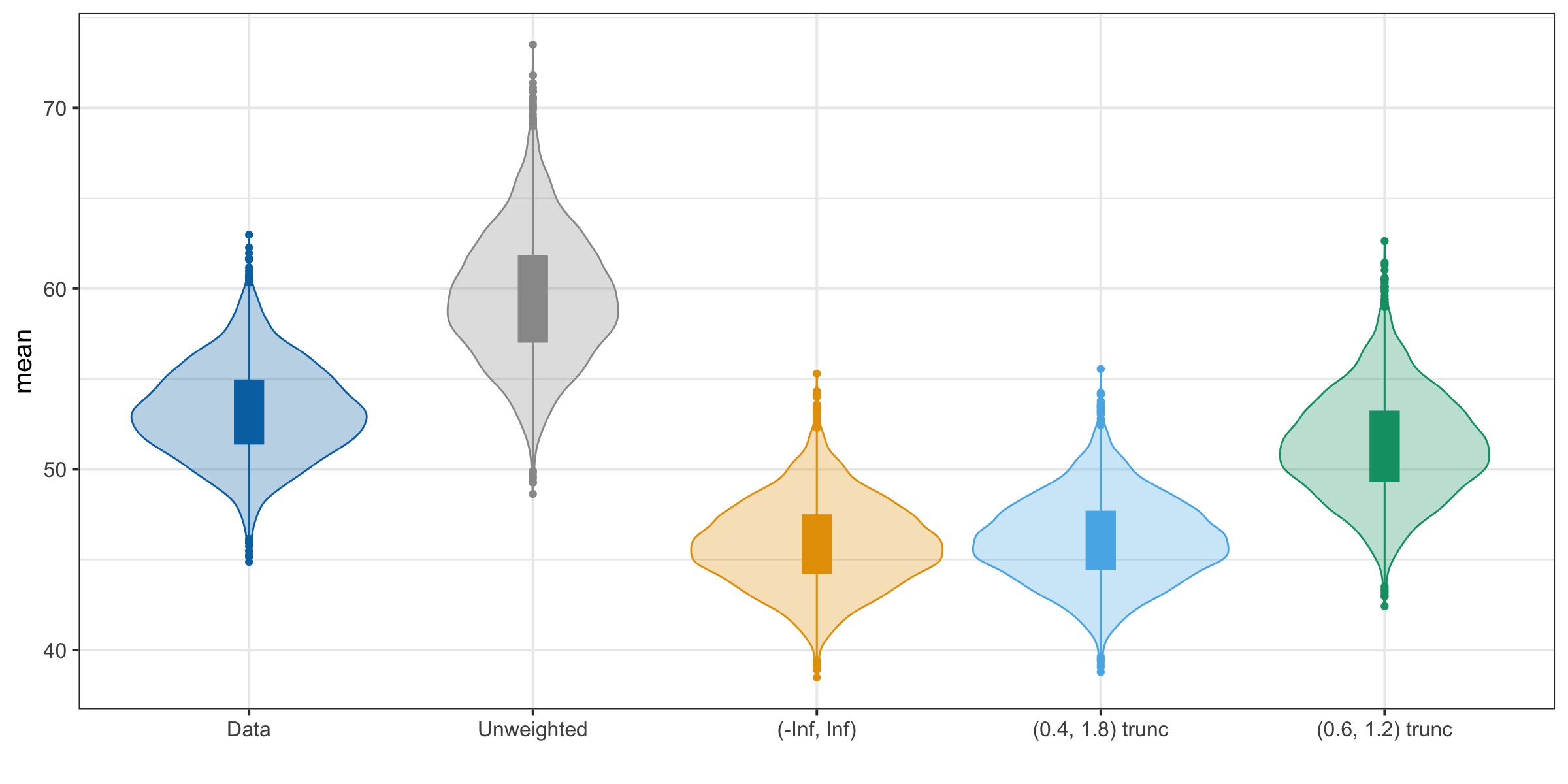}
      \caption{Violin plots of means of data, Unweighted, (-Inf, Inf) (i.e., no bounds as Weighted), $(0.4, 1.8)$ truncated, and $(0.6, 1.2)$ truncated, over a singe sample.}
      \label{fig:sim-mean-single-truncated-similar-eps}
\end{figure}

\begin{figure}[H]
  \centering
   \includegraphics[width=1\textwidth]{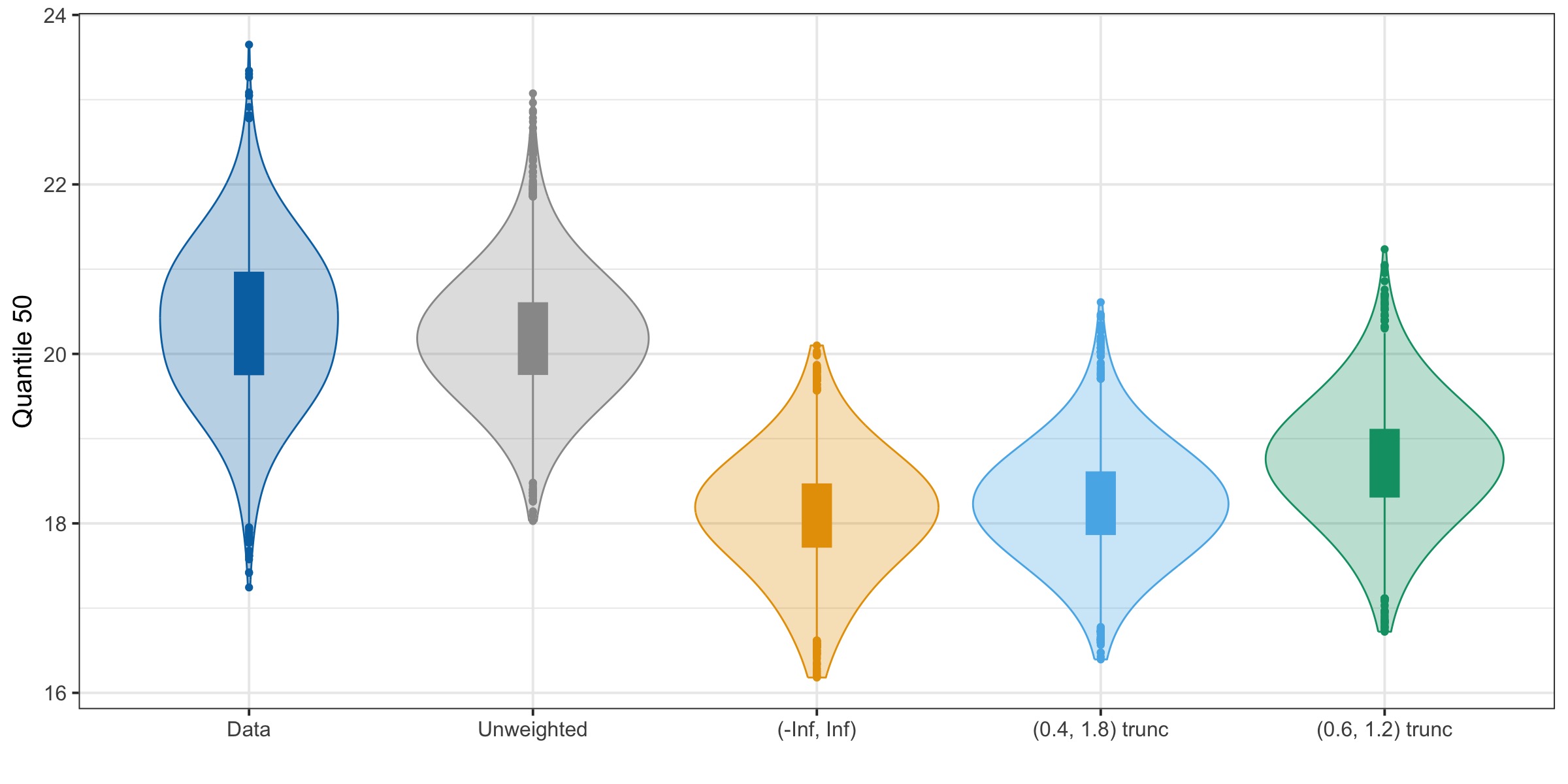}
      \caption{Violin plots of medians of data, Unweighted, (-Inf, Inf) (i.e., no bounds as Weighted), $(0.4, 1.8)$ truncated, and $(0.6, 1.2)$ truncated, over a singe sample.}
      \label{fig:sim-Q50-single-truncated-similar-eps}
\end{figure}

\begin{figure}[H]
  \centering
   \includegraphics[width=1\textwidth]{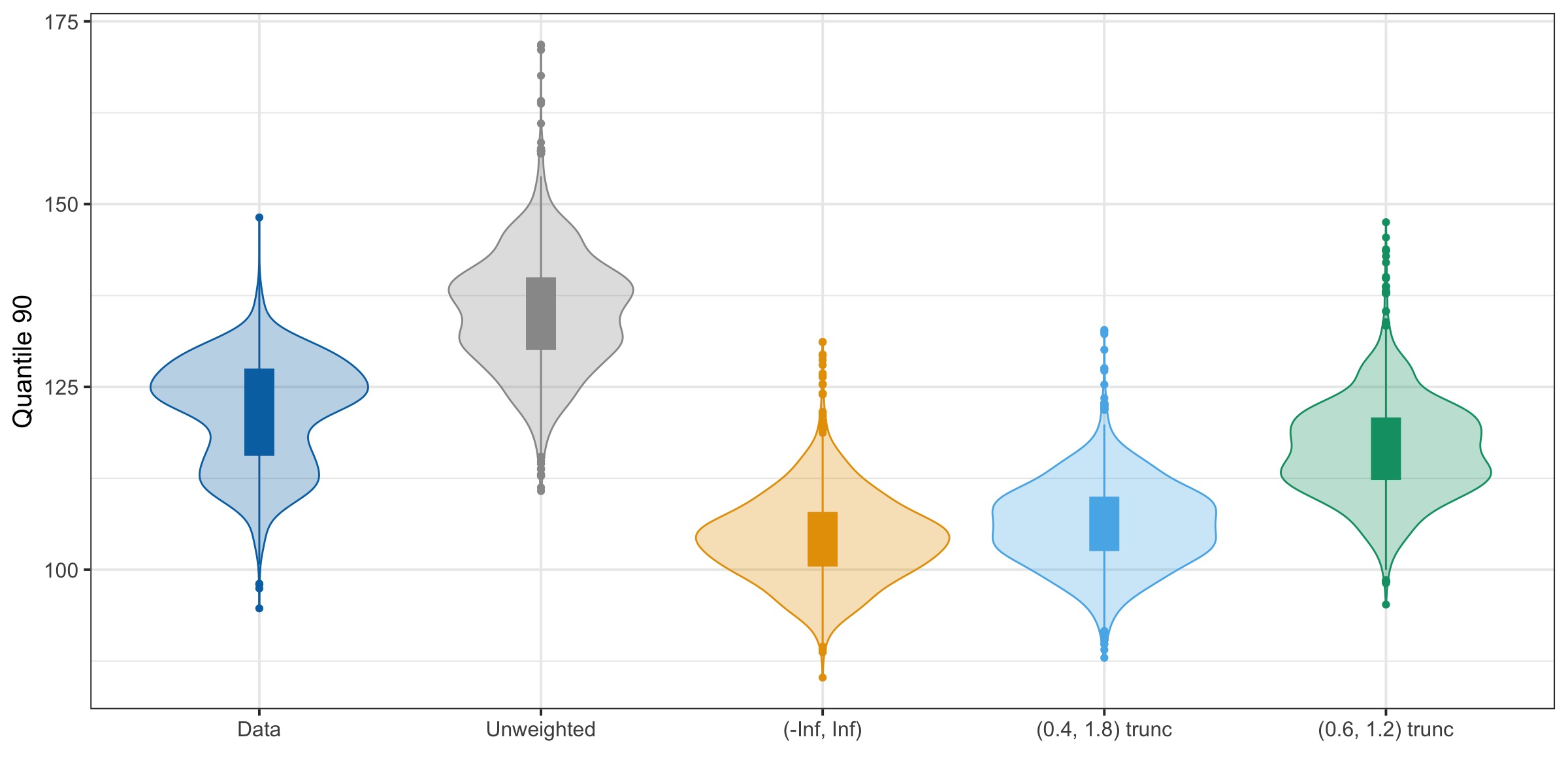}
      \caption{Violin plots of Q90s of data, Unweighted, (-Inf, Inf) (i.e., no bounds as Weighted), $(0.4, 1.8)$ truncated, and $(0.6, 1.2)$ truncated, over a singe sample.}
      \label{fig:sim-Q90-single-truncated-similar-eps}
\end{figure}

\subsection{Assigning wider sensitive range in distribution tail than mode}
\label{sec:simulation:story3}

Data disseminators may want to consider different sensitive ranges for different portions of the confidential data. For example, data disseminators may use a wider sensitive range for records in the right tail of the confidential data distribution, reflecting the need for a higher level of privacy protection of these large, outlying values. Our newly proposed range-restricted standards could easily accommodate such needs and we focus on demonstrating the flexibility through the range-averaged standard in this section.

We continue with our data generating model of $z \sim \textrm{Normal}(2, 1)$ and $x \sim \textrm{Lognormal}(z+1, 1)$ for $n = 2000$ records for 100 repeated samples. In this experiment, we focus on the choice of $(a, b) = (0.4, 1.8)$ for the range-averaged synthesizer, as shown in Section \ref{sec:simulation:story1} and Section \ref{sec:simulation:story2}. However, for large, outlying records in the right tail, we consider $(a', b') = (0.2, 2.4)$, a wider sensitive range than the usual $(a, b) = (0.4, 1.8)$ range, accommodating the needs of higher privacy protection for these records. We consider three different scenarios of large, outlying records receiving this wider sensitive range: top 1\%, top 5\%, and top 10\%. For Figure \ref{fig:sim-Lbounds-top-comparison} through Figure \ref{fig:sim-Q90s-top-comparison}, we label these synthesizers as \emph{(0.4, 1.8)} for the usual range-averaged synthesizer without special treatment for records in the tail, \emph{Top 1\%} for the synthesizer where top 1\% of the large, outlying records uses $(a', b') = (0.2, 2.4)$, \emph{Top 5\%} for the case of top 5\%, and \emph{Top 10\%} for the case of top 10\%. 

First, we examine the privacy budget comparison results in Figure \ref{fig:sim-Lbounds-top-comparison}. As expected, increasing the number of records receiving a wider sensitive range, from 1\% to 10\%, increases the privacy budget. Nevertheless, the Top 10\% synthesizer still has a lower privacy budget than the Weighted synthesizer, indicating a stronger privacy guarantee. This is largely due to the amplification encoded in the range-averaged standard. The Weighted synthesizer provides protection for the full range of all records, whereas the range-averaged synthesizer has a smaller sensitive range to protect for all records.

\begin{figure}[H]
  \centering
   \includegraphics[width=1\textwidth]{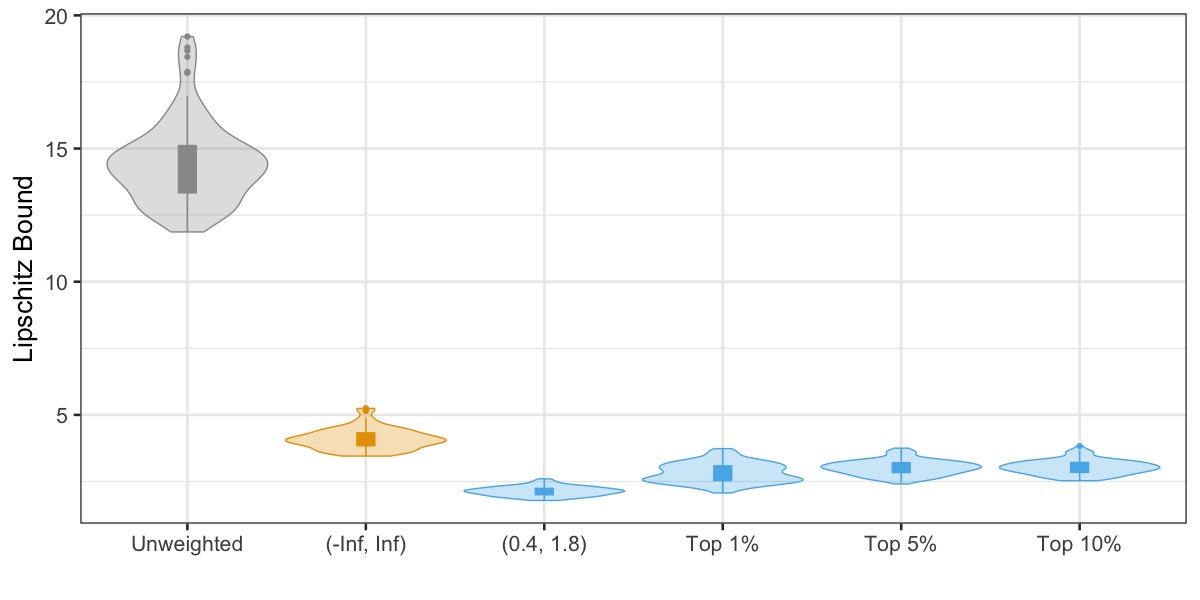}
      \caption{Violin plots of Lipschitz bounds of Unweighted, (-Inf, Inf) (i.e., no bounds as Weighted), $(0.4, 1.8)$ averaged, $(0.4, 1.8)$ averaged with Top 1\% of distribution receive wider sensitivity range of $(0.2, 2.4)$ , $(0.4, 1.8)$ averaged with Top 5\% of distribution receive wider sensitivity range of $(0.2, 2.4)$, and $(0.4, 1.8)$ averaged with Top 10\% of distribution receive wider sensitivity range of $(0.2, 2.4)$, over 100 repeated samples.}
      \label{fig:sim-Lbounds-top-comparison}
\end{figure}

For utility, recall in Section \ref{sec:simulation:story1}, a wider sensitive range in the range-averaged synthesizer provides higher privacy protection and hence introduces more down-weighting and results in lower utility. We therefore observe the wider sensitive range (0.4, 1.8) produces synthetic data with lower utility compared to the shorter sensitive range (0.6, 1.2), for both global utility and analysis-specific utility. Now when providing higher privacy protection to more records in the right tail with a wider sensitive range, from top 1\% to top 10\%, we expect to see more down-weighting and lower utility.

Recall the ECDF maximum record-level difference global utility metric, where higher values indicate lower utility. Figure \ref{fig:sim-ums-top-comparison} shows that as the top percentage value increases from 1\% to 5\% and then to 10\%, the utility decreases as the maximum record-level difference values go up. Similar results of the ECDF average record-level difference global utility metric are included in Appendix \ref{sec:appendixA} for further reading.

For analysis-specific utility comparison, we focus on the medians and 90\% quantile in Figure \ref{fig:sim-medians-top-comparison} and Figure \ref{fig:sim-Q90s-top-comparison} (results of the means are in Appendix \ref{sec:appendixA}). Once again, we observe decreasing utility as the top percentage value increases: the violin plot deviates further from the confidential data as the top percentage value increases from 1\% to 5\% and then to 10\%.

\begin{figure}[H]
  \centering
   \includegraphics[width=1\textwidth]{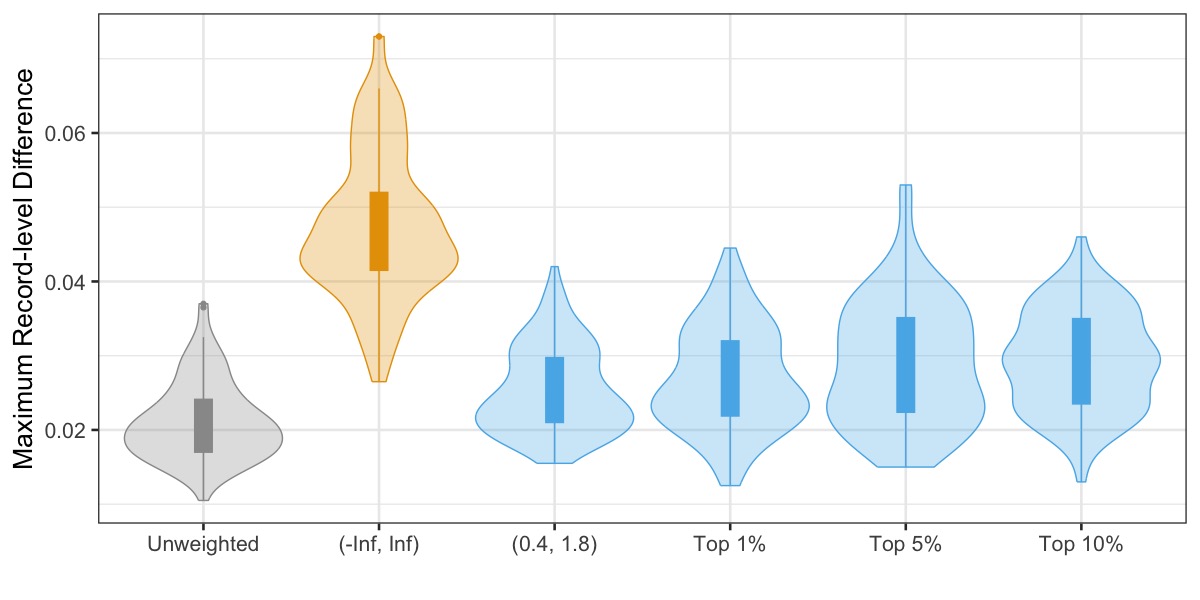}
      \caption{Violin plots of ECDF maximum record-level differences of synthetic data of Data, Unweighted, (-Inf, Inf) (i.e., no bounds as Weighted), $(0.4, 1.8)$ averaged, $(0.4, 1.8)$ averaged with Top 1\% of distribution receive wider sensitivity range of $(0.2, 2.4)$ , $(0.4, 1.8)$ averaged with Top 5\% of distribution receive wider sensitivity range of $(0.2, 2.4)$, and $(0.4, 1.8)$ averaged with Top 10\% of distribution receive wider sensitivity range of $(0.2, 2.4)$, over 100 repeated samples.}
      \label{fig:sim-ums-top-comparison}
\end{figure}

\begin{figure}[H]
  \centering
   \includegraphics[width=1\textwidth]{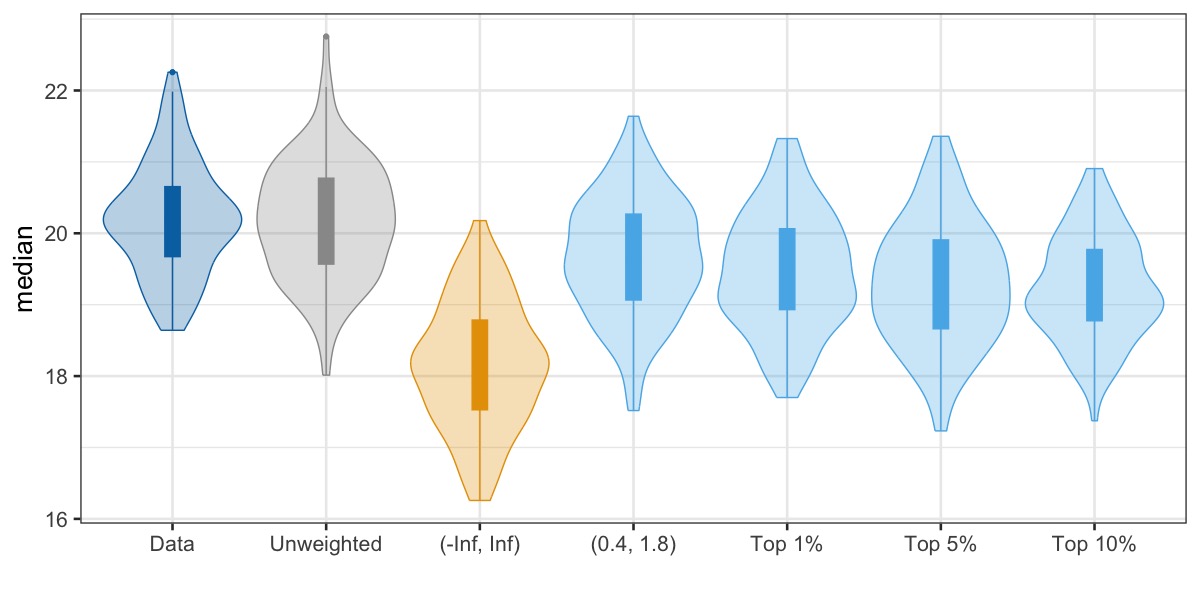}
      \caption{Violin plots of medians of synthetic data of Data, Unweighted, (-Inf, Inf) (i.e., no bounds as Weighted), $(0.4, 1.8)$ averaged, $(0.4, 1.8)$ averaged with Top 1\% of distribution receive wider sensitivity range of $(0.2, 2.4)$ , $(0.4, 1.8)$ averaged with Top 5\% of distribution receive wider sensitivity range of $(0.2, 2.4)$, and $(0.4, 1.8)$ averaged with Top 10\% of distribution receive wider sensitivity range of $(0.2, 2.4)$, over 100 repeated samples.}
      \label{fig:sim-medians-top-comparison}
\end{figure}

\begin{figure}[H]
  \centering
   \includegraphics[width=1\textwidth]{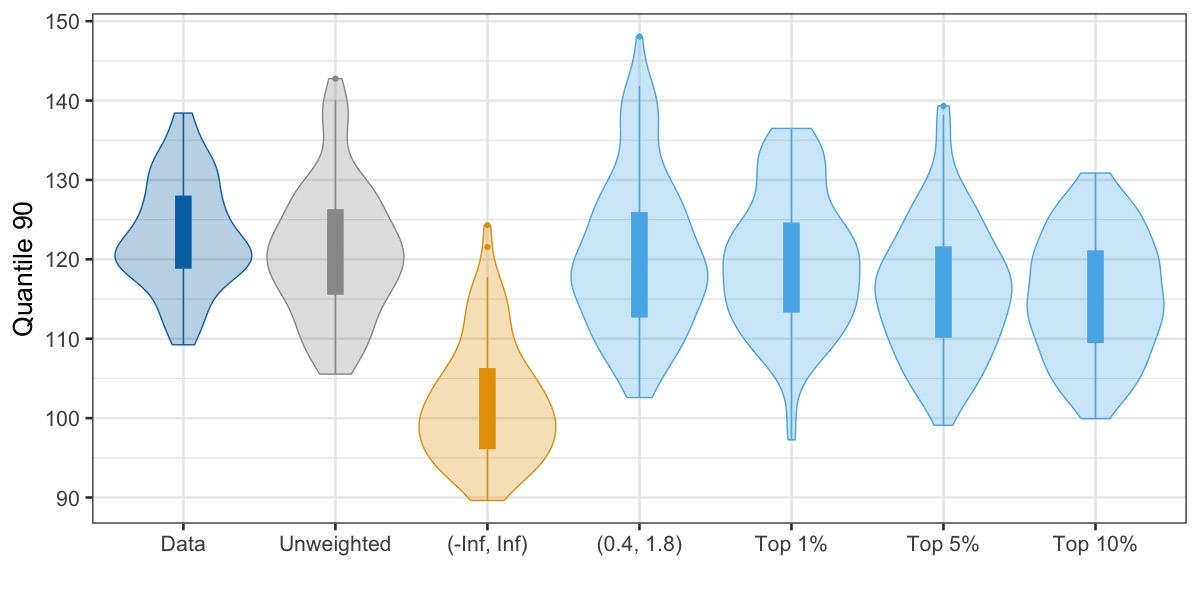}
      \caption{Violin plots of 90th quantiles of synthetic data of Data, Unweighted, (-Inf, Inf) (i.e., no bounds as Weighted), $(0.4, 1.8)$ averaged, $(0.4, 1.8)$ averaged with Top 1\% of distribution  receive wider sensitivity range of $(0.2, 2.4)$ , $(0.4, 1.8)$ averaged with Top 5\% of distribution receive wider sensitivity range of $(0.2, 2.4)$, and $(0.4, 1.8)$ averaged with Top 10\% of distribution receive wider sensitivity range of $(0.2, 2.4)$, over 100 repeated samples.}
      \label{fig:sim-Q90s-top-comparison}
\end{figure}

These privacy and utility results further illustrate the flexibility of our newly proposed range-restricted standards. Data disseminators could consider different sensitive ranges for different portions of the confidential data. In order to provide higher privacy protection for records in the right tail, data disseminators can use a wider sensitive range for these records. This essentially reduces the amount of privacy amplification offered by our range-restricted standards compared to a narrower sensitive range and results in a less amplified privacy guarantee with some reduction in utility.

\section{Application to an accelerated life testing dataset}
\label{sec:application}

We now turn to a real application of accelerated life testing  using Bayesian lognormal regression. The dataset called \texttt{fatigue} is available in the \texttt{BayesLN} R package \citep{BayesLN} and was also considered by \cite{FabriziTrivisano2016SJS}. The dataset consists of the number of test cycles (outcome variable) and a stress factor (a potential predictor variable) for a sample of specimens. We consider the outcome variable cycle and the predictor variable stress, both on the linear and logarithmic scale, with a lognormal regression.

As with the simulation studies, for both range-averaged and range-truncated standards, we use $(a, b) = \{(0.4, 1.8), (0.6, 1.2)\}$, for two sets of sensitive bounds to represent the sensitive range information that the data disseminator attempts to provide privacy protection to. We use $S = 1000$ number of values generated to calculate $\lambda_i$ for the range-averaged standard. We focus on demonstrating that our newly proposed range-restricted synthesizers will strengthen the privacy guarantee, as we have seen in the simulaiton results in Section \ref{sec:simulation:story1}. 

Figure \ref{fig:app-Lbounds} shows the by-record Lipchitz values of the Unweighted synthesizer, the Weighted synthesizer (labeled as (-Inf, Inf) to represent no bounds of $(a, b)$), the two range-averaged synthesizers (labeled as (0.4, 1.8) avg and (0.6, 1.2) avg), and the two range-truncated synthesizers (labeled as (0.4, 1.8) trunc and (0.6, 1.2) trunc)). The privacy budget, $\epsilon_{\mathbf{x}}$, for each synthesizer is calculated as twice of the maximum Lipschitz bound. 
As expected, the two sets of range-restricted synthesizers produce stronger privacy guarantees (lower privacy budgets) compared to the Weighted, by providing focused privacy protection to the sensitive range. For both sets, shortening the sensitive range bounds from (0.4, 1.8) to (0.6, 1.2) provides even stronger privacy guarantees (lower privacy budgets) given a shorter sensitive range for privacy protection. 

 Given the same sensitive range bounds $(a, b)$, the range-averaged synthesizer produces further privacy budget decrease (stronger privacy guarantee) compared to the range-truncated synthesizer. These results are consistent wit the mathematical formulations of the approaches in Section \ref{sec:averaged} and Section \ref{sec:truncated} and from the simulation study results in Section \ref{sec:simulation:story1}. The two sets of range-restricted synthesizers use information from the sensitive range in different ways: The range-averaged synthesizer uses distributional information by computing probabilities of falling within a sensitive range (i.e., more information) while the range-truncated synthesizer uses only the end points of the sensitive range (i.e., less information).

\begin{figure}[H]
  \centering
   \includegraphics[width=1\textwidth]{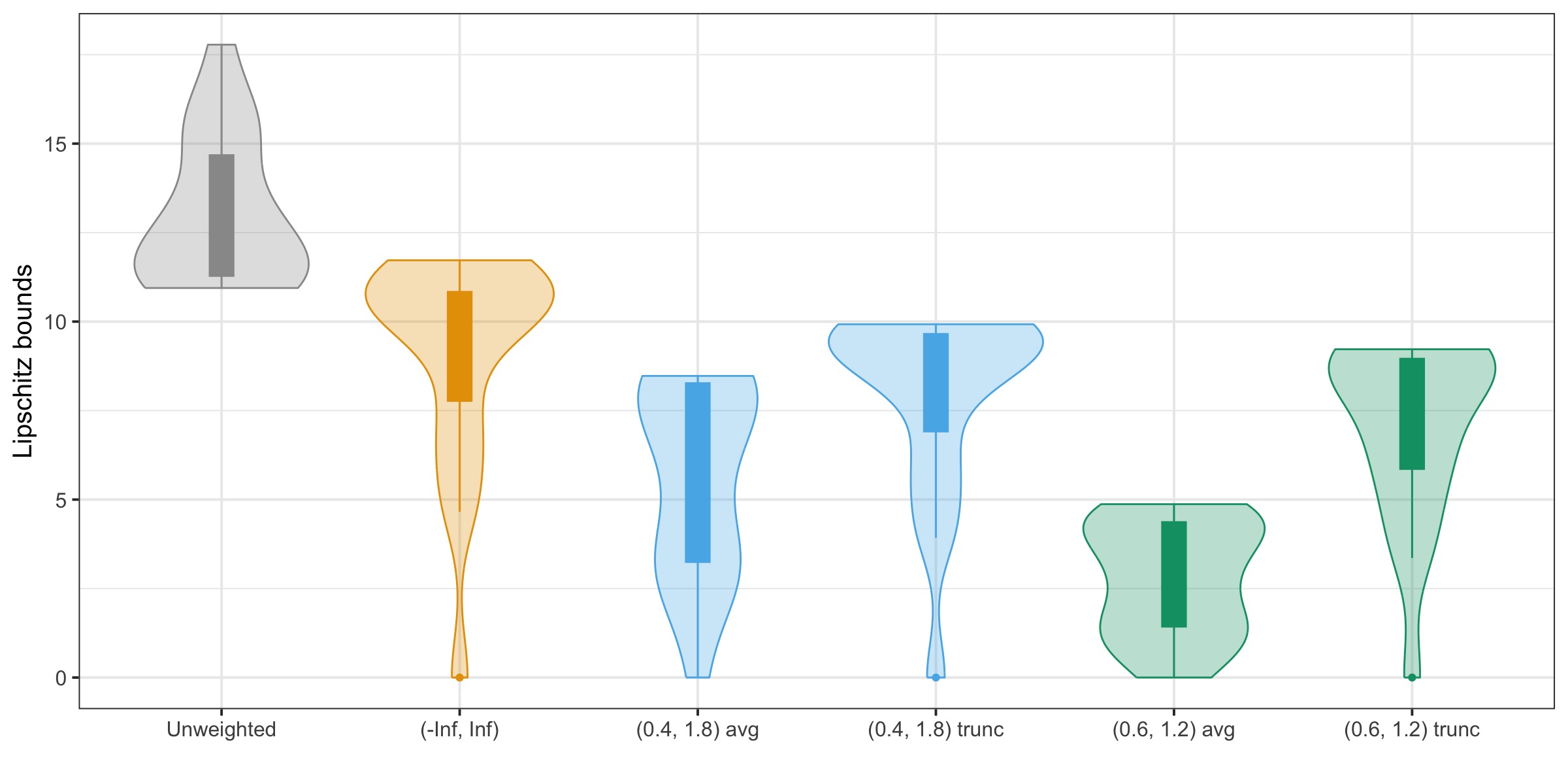}
      \caption{Violin plots of by-record Lipschitz bounds of Unweighted, (-Inf, Inf) (i.e., no bounds as Weighted), $(0.4, 1.8)$ averaged, $(0.4, 1.8)$ truncated, $(0.6, 1.2)$ averaged, and $(0.6, 1.2)$ truncated.}
     \label{fig:app-Lbounds}
\end{figure}

Comparing utility results, Figure \ref{fig:app-data} depicts the distribution of the confidential data and that of the simulated synthetic data from the list of synthesizers we consider. Recall that the range-truncated synthesizers are the same as the Weighted synthesizer with Lipschitz bounds calculated differently. Hence the identical synthetic data distribution in Figure \ref{fig:app-data} between (-Inf, Inf) (i.e., no bounds as Weighted), (0.4, 1.8) truncated, and (0.6, 1.2) truncated. The two range-averaged synthesizers preserve the confidential data distribution better than the Weighted and the two range-truncated synthesizers. Between the two ranged averaged approaches, the shorter sensitive range (0.6, 1.2) provided the best utility preservation. This is once again consistent with previous results, as a shorter sensitive range requires privacy protection for a smaller range of potential values to be protected and hence preserves a higher level of utility. 

\begin{figure}[H]
  \centering
   \includegraphics[width=1\textwidth]{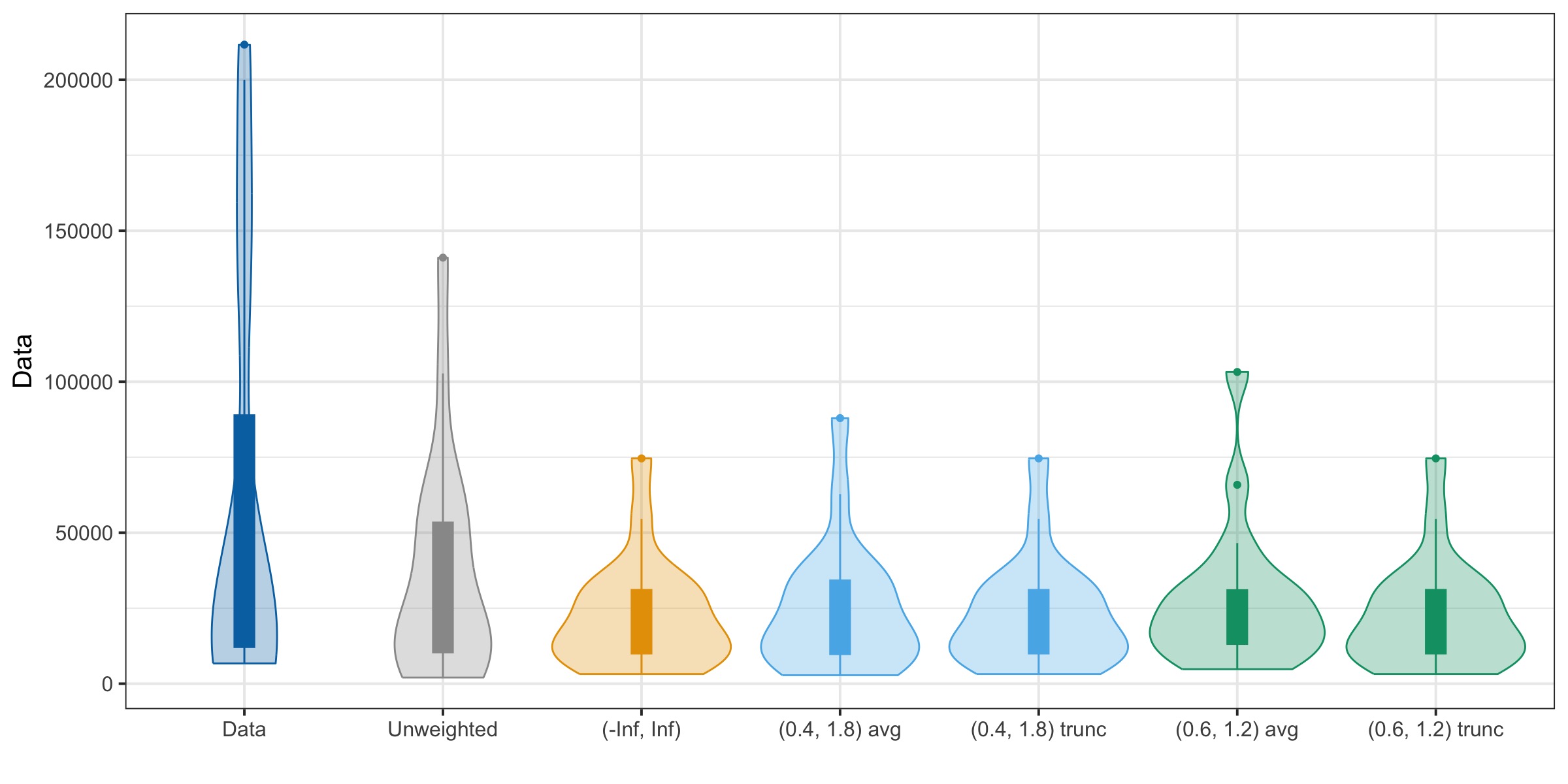}
      \caption{Violin plots of confidential data and synthetic data generated from Unweighted, (-Inf, Inf) (i.e., no bounds as Weighted), $(0.4, 1.8)$ averaged, $(0.4, 1.8)$ truncated, $(0.6, 1.2)$ averaged, and $(0.6, 1.2)$ truncated.}
    \label{fig:app-data}
\end{figure}

Similar results are observed for analysis-specific utility, shown in Table \ref{tab:Stress_app-utility}, that (0.6, 1.2) averaged synthesizer has the highest utility for all three statistics, namely the mean, the median, and the 90th quantile statistic, and (0.4, 1.8) averaged synthesizer comes second. Recall that the Weighted synthesizer produces synthetic datasets identical to that of the (0.4, 1.8) truncated and that of the (0.6, 1.2) truncated synthesizers. These synthetic data have the lowest utility for all three statistics. 

\begin{table}[h!]
\centering
\begin{tabular}{ l | r | r r r r r }
   & Unweighted & (-Inf, Inf) & (0.4, 1.8) avg &(0.4, 1.8) trunc & (0.6, 1.2) avg  & (0.6, 1.2) trunc \\ \hline
 Mean &  36413 & 22743 & {\underline{24762}} & 22743& {\bf{27795}} &  22743 \\
 Median  & 20832 & 15991 & {\underline{16225}} & 15991 & {\bf{19951}} & 15991  \\
 90th Q  & 73446 & 37595  & {\underline{42168}}  & 37595  & {\bf{50424}} &  37595 \\ \hline
\end{tabular}
\caption{Utility results of Unweighted, (-Inf, Inf) (i.e., no bounds as Weighted), $(0.4, 1.8)$ avg, $(0.4, 1.8)$ truncated, $(0.6, 1.2)$ avg, and $(0.6, 1.2)$ truncated, over a singe sample. In each utility row, the best performing synthesizer is in bold and the second best is underlined. The mean, median, and 90th quantile from the confidential data are 57771, 15616, and 165385, respectively.}
\label{tab:Stress_app-utility}
\end{table}

Overall, our application results are consistent with our findings in the simulation study in Section \ref{sec:simulation:story1}, that both the range-averaged and the range-truncated synthesizers provide stronger privacy guarantee than the Weighted synthesizer without considering a sensitive range.  For either, a shorter sensitive range results in a stronger privacy guarantee. Between the two, the range-averaged synthesizer provides an even stronger privacy guarantee by utilizing more information from the selected sensitive range. On the utility side, the range-truncated synthesizers produce synthetic data identical to that of the Weighted synthesizer, hence identical utility profiles. The range-averaged synthesizers produce synthetic data with higher utility, with a shorter sensitive range resulting in even higher utility. Additional results for the performance of the synthesizers, including the values of $\boldsymbol{\lambda}$, the sensitive information probabilities for the range-averaged synthesizers and the privacy weights $\boldsymbol{\alpha}$ of all synthesizers, are available in Appendix \ref{sec:appendixB}.

\section{Concluding remarks}
\label{sec:conclusion}
In this work, we propose a general formulation for model-based privacy protection restricted to a known range of data values. This allows for the use of the risk-weighted pseudo-posterior mechanism along with outside or public knowledge to provide more targeted and efficient formal privacy protection. Our general approach is based on a decomposition of the individual likelihood contributions to risk and results in two novel range-restricted standards: range-averaged and range-truncated.
Each standard utilizes information from a sensitive range of values in different ways. The range-averaged approach leverages more information about the data distribution within the restricted range, resulting in significant improvements in both privacy and utility. The range-truncated approach is more conservative and only conditions on information from the endpoints of the restricted range. This results in more modest gains in privacy over the regular risk-weighted pseudo-posterior. However, both provide a privacy amplification effect leading to stronger, more focused protection for values within the specified range.

Our series of simulation studies in Section \ref{sec:simulation} demonstrate the privacy amplification effects of our two new range-restricted standards. We also show the role that the length of the sensitive range $(a, b)$ plays in tuning the privacy budget and the utility preservation level. In Section \ref{sec:simulation:story1}, both sets of range-restricted synthesizers produce a stronger and amplified privacy guarantee by utilizing  information about sensitive ranges possessed by data disseminators without a compromise of data utility. The length of the sensitive range $(a, b)$ plays an important role when it comes to tuning the amount of amplification. Specifically, shortening the sensitive range enhances the privacy amplification effect and improves utility. Section \ref{sec:simulation:story2} demonstrates the flexibility of our new range-restricted standards through the scaling of privacy weights $\alpha_i$'s, in order to achieve higher data utility for the same privacy budget. This could be particularly appealing to data disseminators if they would like to utilize the sensitive information they possess and have a particular targeted privacy budget in mind. The flexibility of our new range-restricted standards to allow different sensitive ranges for different portions of the confidential data is illustrated in Section \ref{sec:simulation:story3}. Specifically, data disseminators may choose a wider sensitive range for large, outlying records, in order to provide higher privacy protection. Our range-restricted standards can easily incorporate such choices into the synthesizer. The resulting synthetic data would have a larger privacy budget and lower data utility due to reduced privacy amplification through the use of a wider sensitive range. Our application to an accelerated life testing dataset in Section \ref{sec:application} provides another illustration of the privacy and utility profiles of the two newly proposed range-restricted standards. 

We provide a further discussion on the seemingly counterintuitive results of the privacy-utility tradeoff. A widely observed phenomenon in the privacy research literature, especially among additive noise methods, is that a lower privacy budget means adding more noise and hence reducing data utility. \cite{HuSavitskyWilliams2022JSSAM} demonstrate that this is also true in model-based formal privacy. In our results, we do not observe this privacy-utility tradeoff: Section \ref{sec:simulation:story1} shows our range-restricted synthesizers can achieve a lower privacy budget without compromising data utility; Section \ref{sec:simulation:story2} shows we can achieve higher data utility for the same privacy budget through scaling of privacy weights; Section \ref{sec:simulation:story3} shows that we can achieve a higher privacy budget and lower data utility by using a wider sensitive range for large, outlying records. 
While this may seem counterintuitive, we highlight that privacy amplification is the key concept in our new range-restricted standards. Data disseminators possess information about sensitive ranges that they could use for their advantage when it comes to data dissemination. In our range-restricted standards, we design synthesizers that condition on such `public' information, and the enhanced, stronger privacy guarantee (i.e., lower privacy budget) is achieved through the amplification effect of conditioning on this information. Hence, it is possible to reduce privacy budget through the amplification effect without reducing data utility.

\newpage
\bibliographystyle{agsm}
\bibliography{DPbib,A_reference}

\newpage
\appendix
\section{Additional utility results from Section \ref{sec:simulation:story3}}
\label{sec:appendixA}

\begin{figure}[H]
  \centering
   \includegraphics[width=1\textwidth]{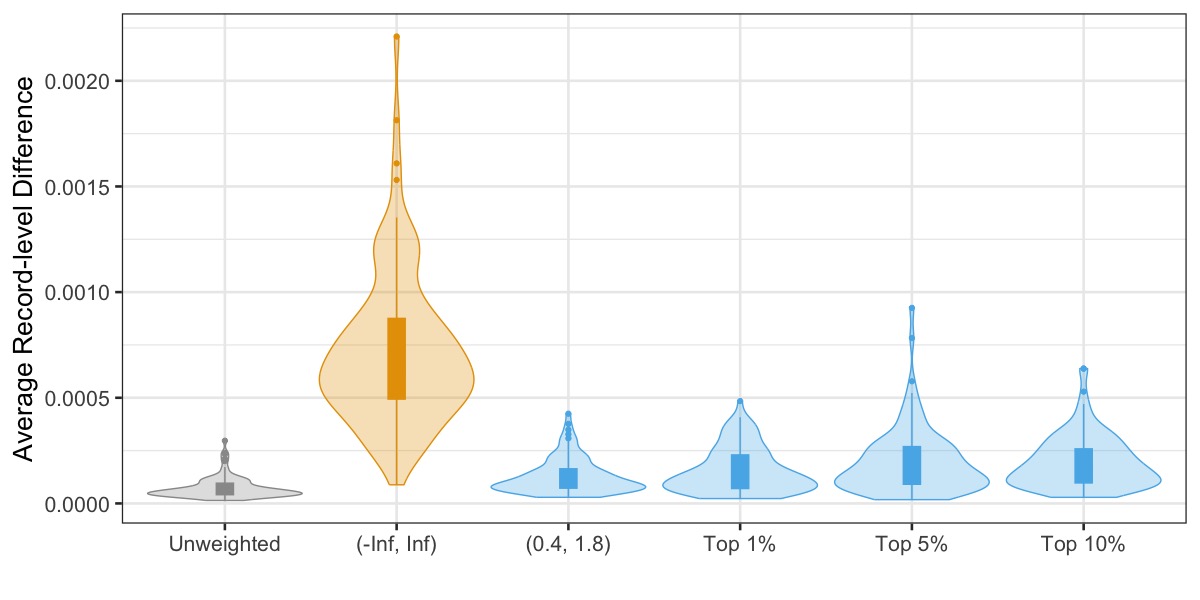}
      \caption{Violin plots of ECDF average record-level differences of synthetic data of Data, Unweighted, (-Inf, Inf) (i.e., no bounds as Weighted), $(0.4, 1.8)$ averaged, $(0.4, 1.8)$ averaged with Top 1\% of distribution  receive wider sensitivity range of $(0.2, 2.4)$ , $(0.4, 1.8)$ averaged with Top 5\% of distribution receive wider sensitivity range of $(0.2, 2.4)$, and $(0.4, 1.8)$ averaged with Top 10\% of distribution receive wider sensitivity range of $(0.2, 2.4)$, over 100 repeated samples.}
      \label{fig:sim-uas-top-comparison}
\end{figure}

\begin{figure}[H]
  \centering
   \includegraphics[width=1\textwidth]{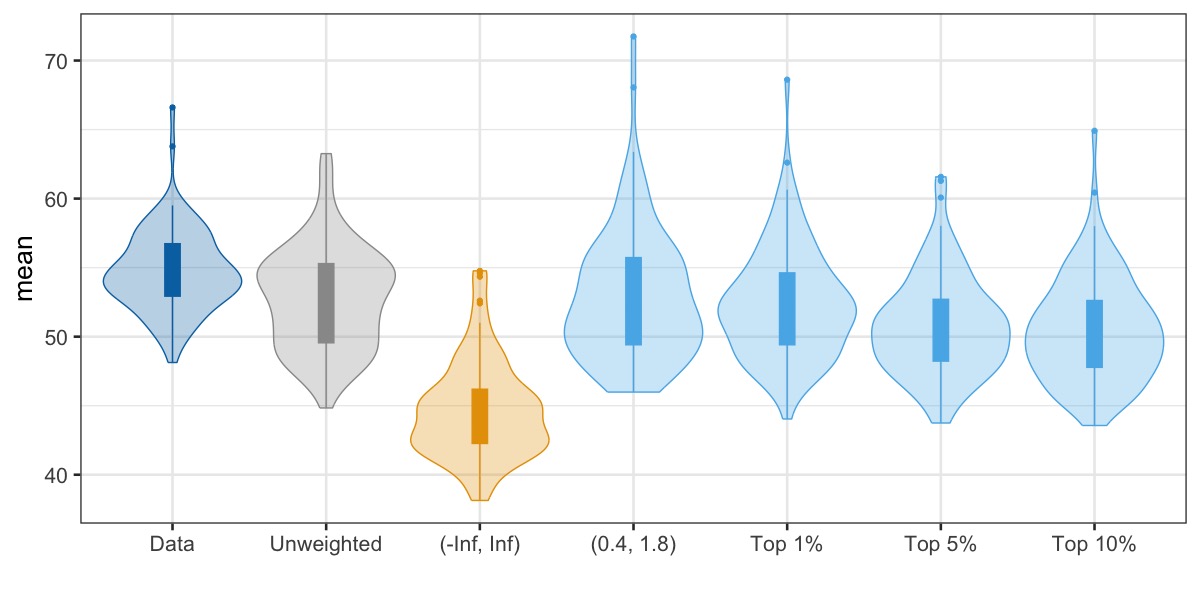}
      \caption{Violin plots of means of synthetic data of Data, Unweighted, (-Inf, Inf) (i.e., no bounds as Weighted), $(0.4, 1.8)$ averaged, $(0.4, 1.8)$ averaged with Top 1\% of distribution  receive wider sensitivity range of $(0.2, 2.4)$ , $(0.4, 1.8)$ averaged with Top 5\% of distribution receive wider sensitivity range of $(0.2, 2.4)$, and $(0.4, 1.8)$ averaged with Top 10\% of distribution receive wider sensitivity range of $(0.2, 2.4)$, over 100 repeated samples.}
      \label{fig:sim-means-top-comparison}
\end{figure}

\section{Additional plots from Section \ref{sec:application}}
\label{sec:appendixB}

\begin{figure}[H]
  \centering
   \includegraphics[width=1\textwidth]{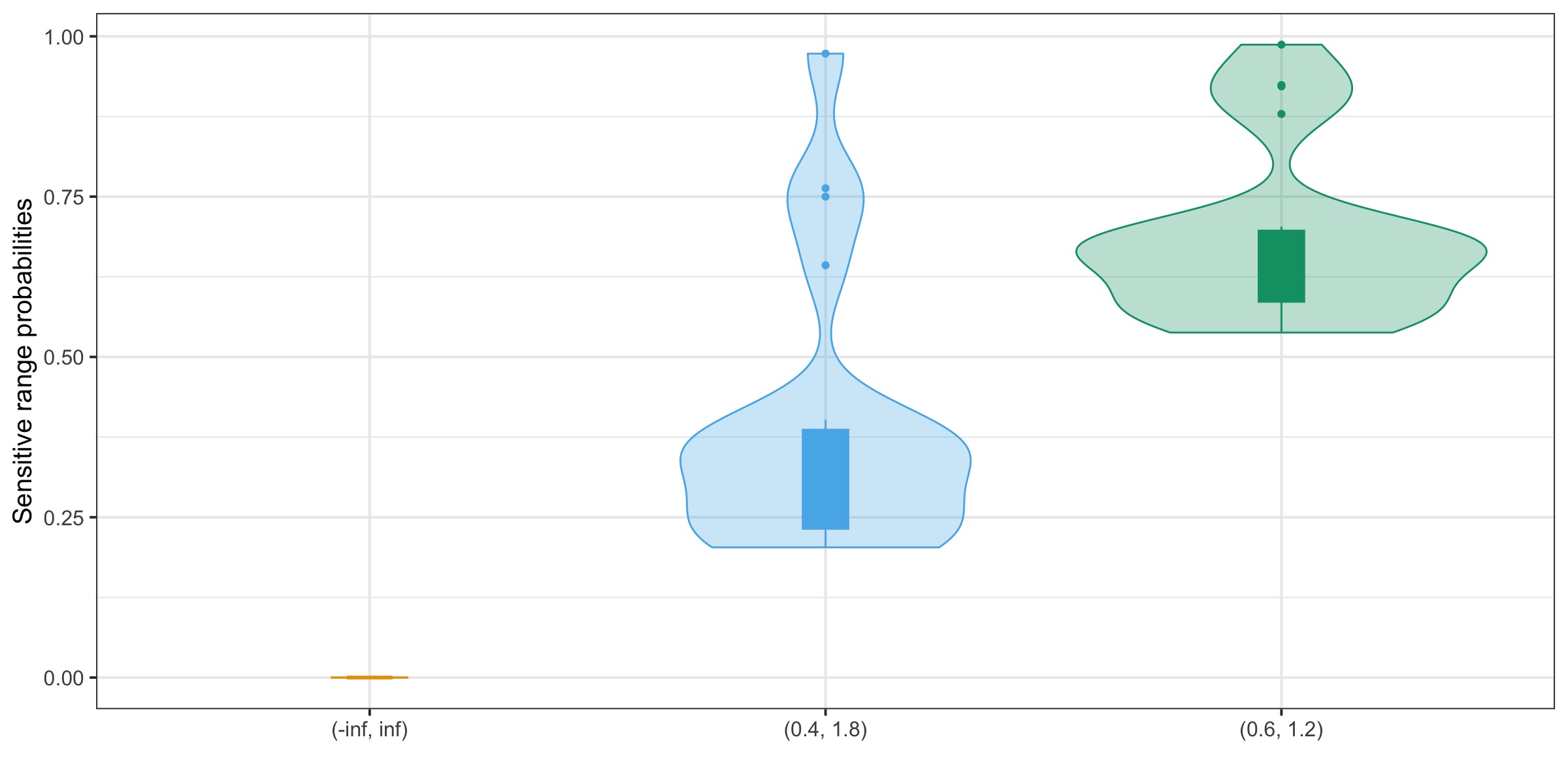}
      \caption{Violin plots of sensitive range probabilities (i.e., $\lambda$) of (-Inf, Inf) (i.e., no bounds as Weighted), $(0.4, 1.8)$ averaged, and $(0.6, 1.2)$ averaged.}
      \label{fig:app-lambdas}
\end{figure}

\begin{figure}[H]
  \centering
   \includegraphics[width=1\textwidth]{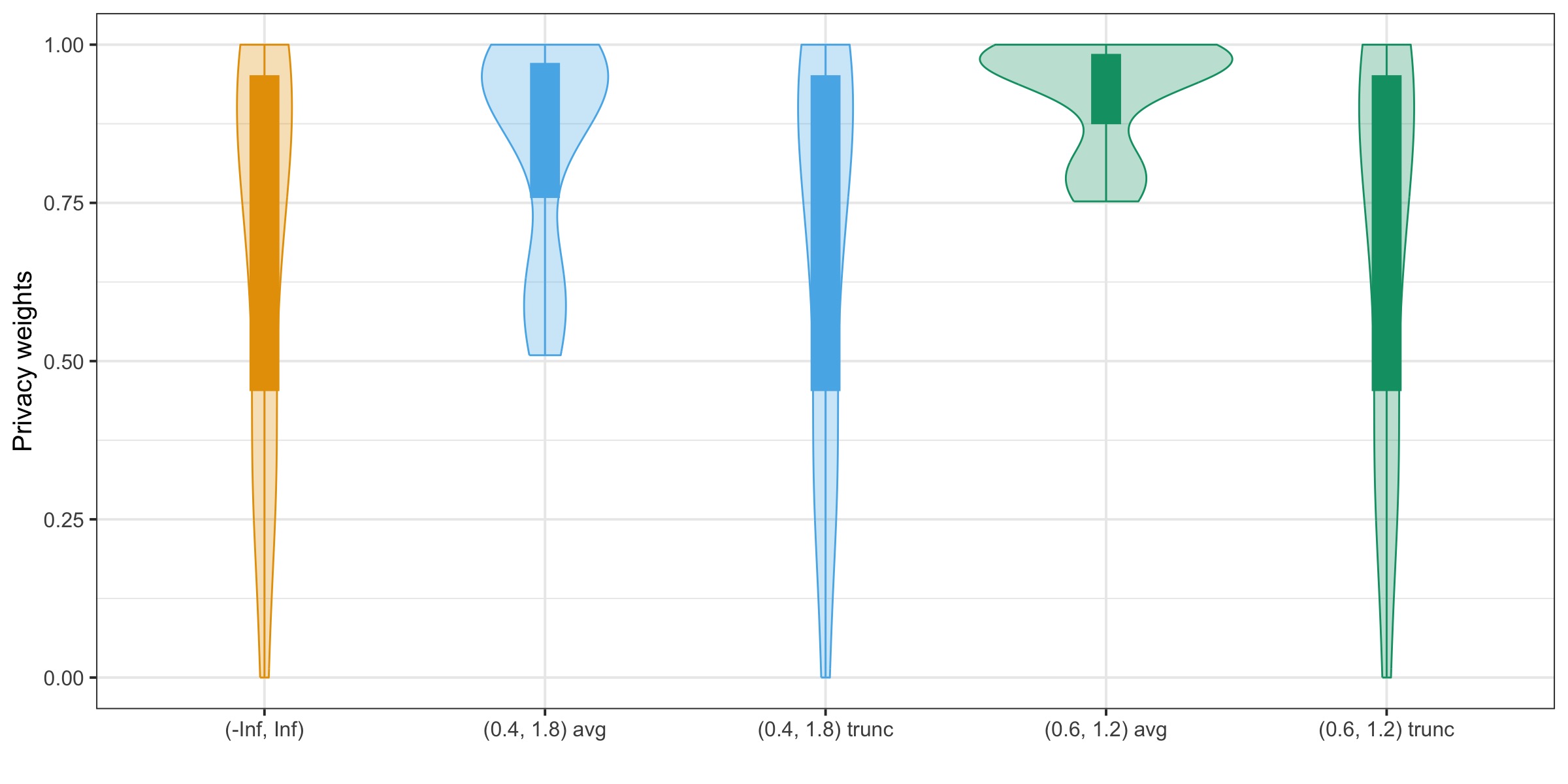}
      \caption{Violin plots of privacy weights of (-Inf, Inf) (i.e., no bounds as Weighted), $(0.4, 1.8)$ averaged, $(0.4, 1.8)$ truncated, $(0.6, 1.2)$ averaged, and $(0.6, 1.2)$ truncated.}
      \label{fig:app-alphas}
\end{figure}

\end{document}